\documentclass[11pt,a4paper, reqno]{article}
\usepackage[utf8]{inputenc}

\usepackage[sc]{mathpazo} 
\usepackage{eulervm}
\usepackage{comment}
\usepackage[T1]{fontenc}
\usepackage[dvipsnames]{xcolor}
\usepackage[numbers,sort]{natbib}
\usepackage[affil-it]{authblk}

\usepackage{microtype}
\usepackage{amsfonts, amsmath, amssymb, amsthm, bbm, enumerate, graphicx, tikz, hyperref, enumerate,color, subcaption} 

\newtheorem*{algorithm*}{Algorithm}
\newtheorem{theorem}{Theorem}
\newtheorem{corollary}{Corollary}
\newtheorem{definition}{Definition}
\theoremstyle{plain}
\newtheorem{proposition}[theorem]{Proposition}
\newtheorem{lemma}[theorem]{Lemma}
\newtheorem{claim}{Claim}
\newtheorem{observation}{Observation}
\newcommand{\N}{\mathbbm{N}}
\newcommand{\bip}{\mathrm{bip}}
\newcommand{\fvs}{\mathrm{\textsc{fvs}}}
\newcommand{\vc}{\mathrm{\textsc{vc}}}
\newcommand{\hS}{\hat{S}}

\newcommand{\tree}{\ensuremath{\mathop{\mathrm{\textsc{treeminor}}}}}
\newcommand{\I}{\mathcal{I}}

\newenvironment{claimproof}{\begin{proof}\renewcommand{\qedsymbol}{\claimqed}}{\end{proof}\renewcommand{\qedsymbol}{\plainqed}}

\newtheorem{fact}{Fact}

\let\plainqed\qedsymbol

\usepackage[a4paper,includeheadfoot, total={5.9in,9.5in}]{geometry}

\usepackage{hyperref}
\hypersetup{
  colorlinks,
  linkcolor=blue,
  citecolor=blue,
  filecolor=blue,
  urlcolor=black,
  pdftitle={},
  pdfauthor={},
  pdfcreator={D. Mukherjee},
  pdfsubject={},
  pdfkeywords={}
}

\newcommand{\Pum}{\ensuremath{\mathrm{\textsc{pum}}}}
\newcommand{\bistable}{\ensuremath{\mathrm{\textsc{bi}}}}
\newcommand{\MTJ}{\ensuremath{\mathrm{\textsc{mtj}}}}
\newcommand{\TAR}{\ensuremath{\mathrm{\textsc{tar}}}}
\newcommand{\pw}{\mathrm{\textsc{pw}}}

\newcommand{\depth}{\ensuremath{\mathrm{\textsc{depth}}}}

\begin{document}

\title{Independent Set Reconfiguration Thresholds of Hereditary Graph Classes}

\author{Mark de Berg}
\author{Bart M.P. Jansen}
\author{Debankur Mukherjee}
\affil{Department of Mathematics and Computer Science,

Eindhoven University of Technology, The Netherlands}

\renewcommand\Authands{, and }

\date{\today}

\maketitle

\begin{abstract}
Traditionally, reconfiguration problems ask the question whether a given solution of an optimization problem can be transformed to a target solution in a sequence of small steps that preserve feasibility of the intermediate solutions. In this paper, rather than asking this question from an algorithmic perspective, we analyze the combinatorial structure behind it. We consider the problem of reconfiguring one independent set into another, using two different processes: (1)~exchanging exactly~$k$ vertices in each step, or (2)~removing or adding one vertex in each step while ensuring the intermediate sets contain at most~$k$ fewer vertices than the initial solution. We are interested in determining the minimum value of~$k$ for which this reconfiguration is possible, and bound these threshold values in terms of several structural graph parameters. For hereditary graph classes we identify structures that cause the reconfiguration threshold to be large.
\end{abstract}

\section{Introduction}
Over the past decade, reconfiguration problems have drawn a lot of attention of researchers in algorithms and combinatorics \cite{B14,BKW14,DDFHI15,FHOU15,HD05,IDHPSUU11,KMM12,MNRSS13,W14}.
In this framework, one asks the following question: Given two solutions $I$, $J$ of a fixed optimization problem, can $I$ be transformed into $J$ by a sequence of small steps that maintain feasibility for all intermediate solutions?
Such problems are practically motivated by the fact it may be impossible to adapt a new production strategy instantaneously if it differs too much from the strategy that is currently in use; changes have to be made in small steps, but production has to keep running throughout. From a theoretical perspective, the study of reconfiguration problems provides deep insights into the structure of the solution space.
One of the well-studied examples is when the solution space consists of all the independent sets of a graph (optionally all having a prescribed size).
In this case, three types of reconfiguration rules have been considered. These are naturally explained using \emph{tokens} on vertices of the graph. In \emph{Token Addition Removal} (TAR)~\cite{IDHPSUU11, MNRSS13}, there is a token on every vertex of the initial independent set, and there is a buffer of tokens, initially empty. A step consists of removing a token from a vertex and placing it in the buffer, or placing a buffer token onto a vertex of the graph. The set of vertices with tokens must form an independent set at all times, and the goal is to move the tokens from the initial to the target independent set while ensuring the buffer size never exceeds a given threshold. In \emph{Token Sliding} (TS)~\cite{KMM12, HD05}, a step consists of replacing one vertex~$v$ in the independent set by a neighbor of~$v$ (the token slides along an edge). In \emph{Token Jumping} (TJ)~\cite{KMM12} a step also consists of replacing a single vertex, but the newly added vertex need not have any neighboring relation with the replaced vertex (the token jumps). Token jumping reconfiguration is equivalent to TAR reconfiguration with a buffer of size one.

These models have been analyzed in detail in the recent literature on algorithms~\cite{B14,BKW14,DDFHI15,FHOU15,GH10,LokshtanovMPRS15}, complexity theory~\cite{HD05,IDHPSUU11,KMM12,MNRSS13}, combinatorics~\cite{CDP06,DP06}, and even statistical physics~\cite{JLNSW12,KR15,NZB16}. It is known that the reconfiguration problem under all the above three rules is PSPACE-complete for general graphs, perfect graphs, and planar graphs \cite{HD05,KMM12,IDHPSUU11}. The TJ and TAR reconfiguration problems are PSPACE-complete even for bounded bandwidth graphs~\cite{W14}. Further analyses on the complexity can be found in~\cite{B14,BKW14,DDFHI15,FHOU15,LokshtanovMPRS15}. The constrained token moving problems are related to pebbling games that have been studied in the literature, with applications to robot motion planning~\cite{ABHS15,CDP06, DP06, GH10}.

As mentioned, the goal in reconfiguring independent sets is to go from one given independent $I$ to another one $J$ by a sequence of small steps. In the TS and TJ models, a step involves moving only a single token. This is ideal, but unfortunately reconfiguration is often impossible in the TS or TJ model. Reconfiguration in the TAR model is always possible if one makes the buffer size sufficiently large. However, having a large buffer size is undesirable. We are interested in determining the minimum buffer size that is sufficient to ensure any independent set in a given graph~$G$ can be reconfigured to any target independent set of the same size. We call this minimum the TAR \emph{reconfiguration threshold} (precise definitions in Section~\ref{sec:prelim}). Our aim is to bound the threshold in terms of properties of the graph, and to identify the structures contained in hereditary graph classes that cause the thresholds to be large. We also generalize the TJ model to \emph{Multiple Token Jumping} (MTJ), where in each step a prescribed number of tokens may be moved simultaneously. In the MTJ model, the question becomes: What is the minimum number of simultaneously jumping tokens needed to ensure any reconfiguration is possible? This quantity is called the MTJ \emph{reconfiguration threshold}.

\begin{figure}
\begin{subfigure}{4.6cm}
\begin{center}
\includegraphics[scale=0.35]{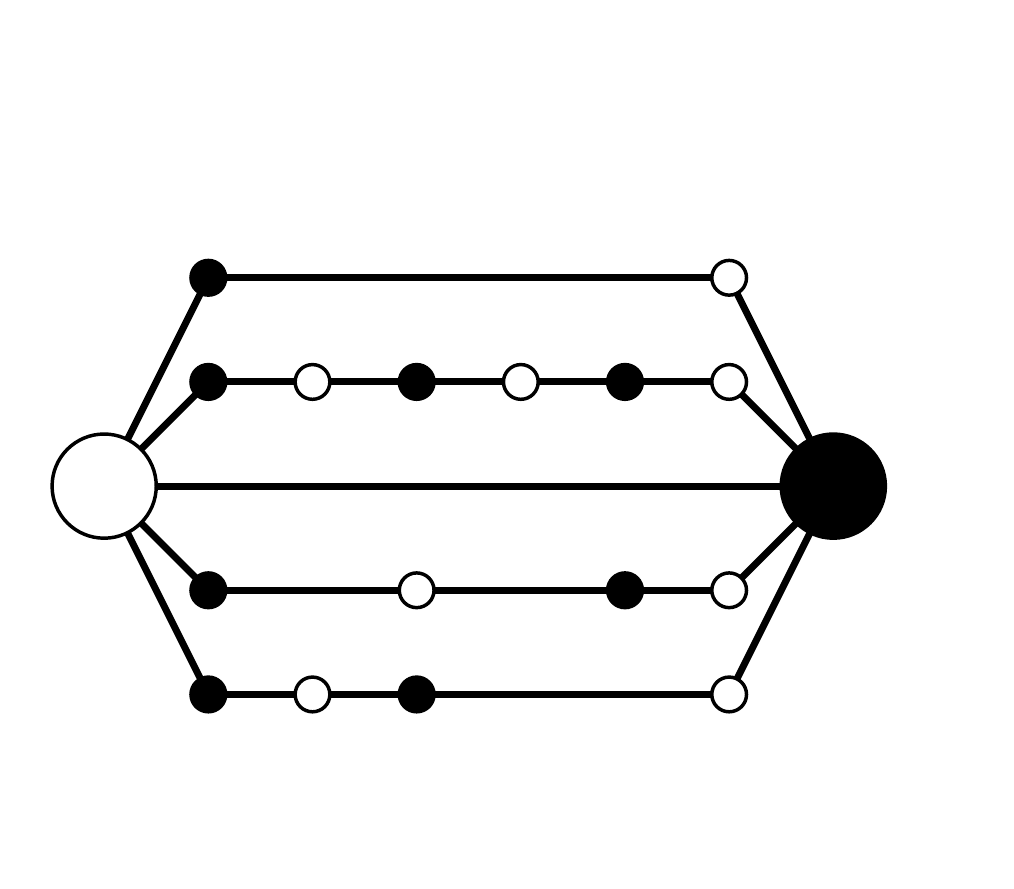}
\end{center}
\caption{A pumpkin of size 18.}
\label{fig:pumpkin}
\end{subfigure}
\begin{subfigure}{10cm}
\begin{center}
\includegraphics[scale=0.35]{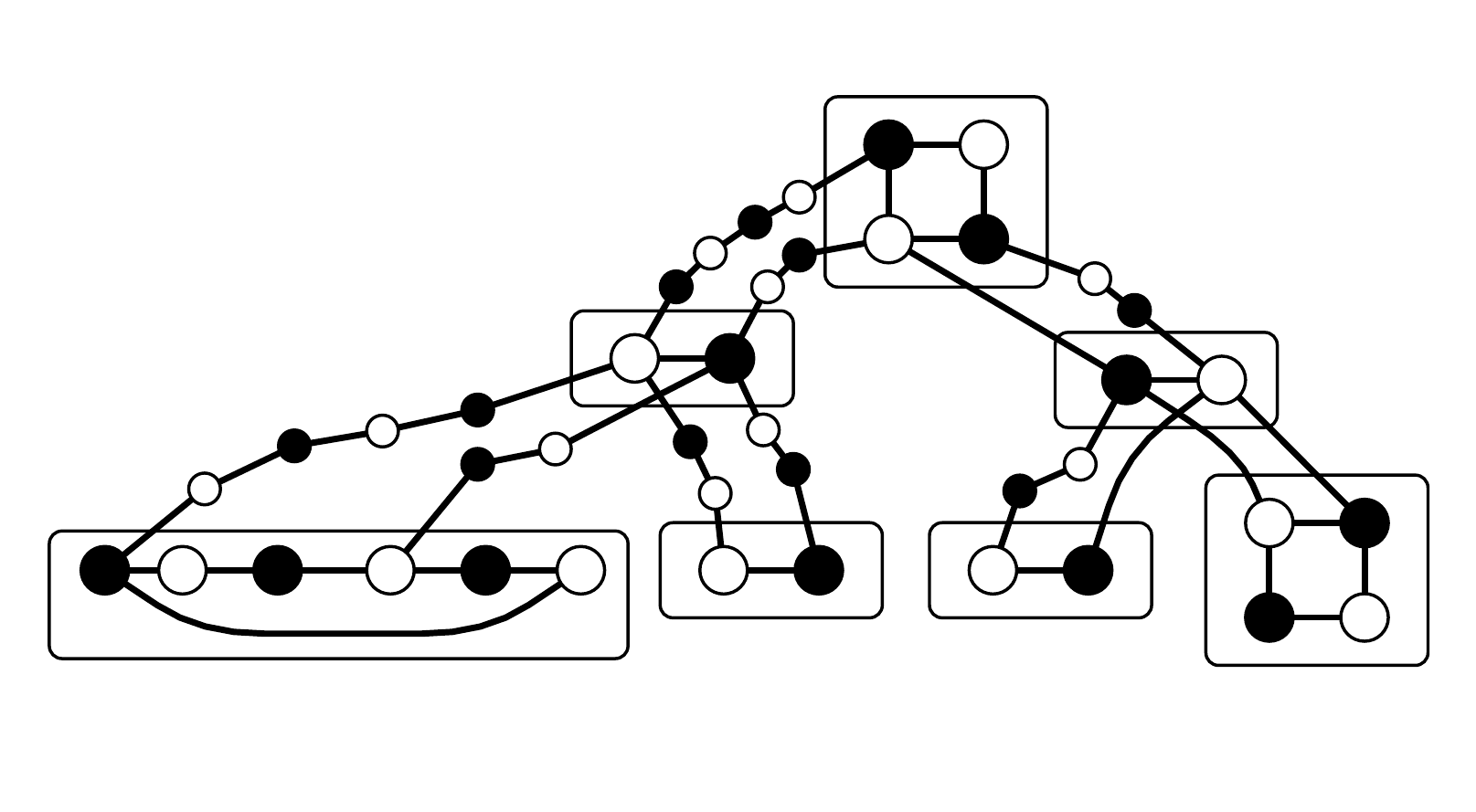}
\end{center}
\caption{A graph of treewidth two with a complete binary tree~$T$ of depth two as a bipartite topological double minor.}
\label{fig:binary}
\end{subfigure}
\caption{The bipartite structures responsible for large MTJ and TAR reconfiguration thresholds, respectively. A \emph{pumpkin} consists of odd-length vertex-disjoint paths between two vertices. The special form of topological \emph{minor} represents each vertex of the tree~$T$ by an edge or even cycle in~$G$, and each edge of~$T$ by two odd-length paths connecting vertices in opposite partite sets in~$G$.}
\end{figure}

\subparagraph*{Our contribution.} We provide upper and lower bounds on the MTJ and TAR reconfiguration thresholds in terms of several graph parameters. In general, our bounds apply to the reconfiguration thresholds of hereditary \emph{graph classes}. The threshold of a graph class is the supremum of the threshold values of the graphs in that class: it is the smallest value~$k$ such that for any graph in the class, any source independent set~$I$ in that graph can be reconfigured into any target independent set~$J$ using steps of size~$k$ (for MTJ) or a buffer of size~$k$ (for TAR).

The MTJ reconfiguration threshold of graphs that are structurally very simple, may nevertheless be very large. For example, an even cycle with~$2n$ vertices can be partitioned into two independent sets~$I$ and~$J$ of size~$n$ each. Any MTJ reconfiguration of~$I$ into~$J$ requires a jump of~$n$ vertices, and this is trivially sufficient. Since a cycle has a feedback vertex set (FVS, see Section~\ref{sec:prelim}) of size one, the MTJ threshold cannot be bounded in terms of the size of a minimum feedback vertex set. However, we prove that the threshold is upper-bounded by the size of a minimum vertex cover of~$G$. Although this bound is tight in the worst case, there are many graph classes with a small MTJ threshold even though they require a large vertex cover. Trees for example have MTJ threshold at most one. We therefore introduce the notion of \emph{pumpkin}, which consists of two nodes connected by at least two vertex-disjoint paths of odd length (Figure~\ref{fig:pumpkin}). The \emph{size} of a pumpkin is the total number of vertices in the structure. We characterize the MTJ reconfiguration threshold of a hereditary graph class~$\Pi$ in terms of the size of the largest pumpkin it contains: the MTJ reconfiguration threshold is upper- and lower-bounded in terms of the largest pumpkin contained in a bipartite graph in~$\Pi$.

TAR reconfiguration is more versatile than MTJ reconfiguration. In the concrete example of a $2n$-cycle discussed above, its MTJ threshold is~$n$ while any pair of independent sets can be reconfigured in the TAR model using a buffer of size two. Moreover, we show that any graph that has a feedback vertex set of size $k$ has TAR reconfiguration threshold at most $k+1$, and reconfiguring one side of the complete bipartite graph~$K_{n,n}$ to the other side shows that this is tight. Our main result concerning TAR reconfiguration states that the TAR reconfiguration threshold of any graph is upper-bounded by its pathwidth. Somewhat surprisingly, there are graphs of constant treewidth (treewidth 2 suffices) for which the TAR reconfiguration threshold is arbitrarily large. We also introduce the concept of \emph{bipartite topological double minor} (BTD-minor), see Figure~\ref{fig:binary}, and show using an isoperimetric inequality that any hereditary graph class containing a graph having a complete binary tree of depth $d$ as a BTD-minor, has TAR reconfiguration threshold $\Omega(d)$. We conjecture that the TAR reconfiguration threshold can also be upper-bounded in terms of the depth of the largest complete binary tree BTD-minor, but we have not been able to prove this (see Section~\ref{sec:conclusion}).

We require the restriction to hereditary graph classes in some of our statements to be able to develop meaningful lower bounds on reconfiguration thresholds, as explained next. Let~$G$ be the disjoint union of~$K_{n,n}$ and a graph~$H$, and let~$I$ and~$J$ be the two partite sets of~$K_{n,n}$. One can verify that~$I$ can be reconfigured to~$J$ by jumps of size at most one if and only if~$H$ has an independent set of size~$n-1$. Similarly,~$I$ can be TAR reconfigured to~$J$ using a buffer of size~$k$ if and only if~$H$ has an independent set of size~$n-k$. Since the size of a maximum independent set is NP-complete to determine, there are no good characterizations of this quantity. When developing lower bounds on the threshold of a hereditary graph class~$\Pi$, this issue disappears since the reconfiguration threshold of any class containing the graph~$G$ above, is at least as high as the threshold of~$H$ (which must be contained in~$\Pi$ if~$G$ is), which is~$n$. The restriction to hereditary graph classes therefore enables us to focus our attention to reconfiguration problems where all vertices in the graph are contained in either the source or target independent set, thereby avoiding the obstacle that the reconfiguration threshold matches the size of a maximum independent set.

\subparagraph*{Applications.}
The MTJ and TAR reconfiguration thresholds play an important role in statistical physics and wireless communication networks. To understand the importance of the TAR reconfiguration threshold, consider the following process:
In a graph $G$, nodes are trying to become \emph{active} (transmit information) at some rate, independently of each other in a distributed manner.
When a potential activation occurs at a node, it can only become active if none of its neighboring nodes are active at that moment (as otherwise the transmissions would interfere). An active node deactivates at some rate independent of the other processes. At any point in time, the set of active nodes in this process forms an independent set of the graph. In statistical physics, this process is known as Glauber dynamics with \emph{hard-core interaction}. This activity process on graphs has many applications in different fields of study.
Loosely speaking, when the activation rate is large, in the long run the above process always tries to stay in a maximum independent set.
For the graphs with more than one maximum independent set, it is interesting to study the time this process takes to reach a target independent set, starting from some specific independent set. This time has been shown to depend crucially upon what we call the TAR reconfiguration threshold of the underlying graph \cite{NZB16}.
In particular, the mixing time of the Glauber dynamics on a graph increases exponentially with its TAR reconfiguration threshold, and hence the Glauber dynamics on the graph is fast mixing if and only if the TAR reconfiguration threshold is small.

The MTJ reconfiguration threshold of a graph~$G$ can be interpreted in the following way. Consider the auxiliary graph, whose vertices correspond to size-$s$ independent sets in $G$ for some fixed~$s$, with an edge between vertices representing sets $I$, $J$ if $|I \setminus J| \leq k$. Then the MTJ reconfiguration threshold is at most~$k$ if and only if this graph auxiliary graph is connected for all~$s$. The MTJ reconfiguration threshold therefore has applications in the parallel Glauber dynamics (PGD) \cite{JLNSW12, KR15}, where the MTJ reconfiguration threshold provides the jump size required to make the underlying Markov process ergodic.

\subparagraph*{Organization.} The succeeding sections are organized as follows. In Section~\ref{sec:prelim} we provide graph-theoretic preliminaries. In Section~\ref{sec:reconfig} we provide a formal description of the two types of reconfiguration. In Section~\ref{sec:mtj} we analyze MTJ reconfiguration. Section~\ref{sec:tar} deals with TAR reconfiguration.

\section{Preliminaries}\label{sec:prelim}

In this section we give the most important graph-theoretic definitions. Notions not defined here can be found in one of the textbooks~\cite{ParAlgo15,D2000}.
A graph is a pair $G=(V,E)$, where $V$ is the set of vertices, and $E$ is the set of edges. We also use~$V(G)$ and~$E(G)$ to refer to the vertex and edge set of~$G$, when convenient. All graphs we consider are finite, simple, and undirected. For~$U \subseteq V$ we denote by~$G-U$ the graph obtained from~$G$ by removing the vertices in~$U$ and their incident edges.
A set $U\subseteq V$ is called an \emph{independent set} of $G$, if $\{u,v\}\notin E$ for any $u,v\in U$.
The \emph{symmetric difference} of two sets~$U$ and~$U'$ is~$U \Delta U' := (U_1\setminus U_2)\cup (U_2\setminus U_1)$. A set $U\subseteq V$ is a \emph{vertex cover} of $G$ if every edge in $E$ is incident with a vertex in $U$.
The minimum cardinality of a vertex cover of $G$ is denoted by $\vc(G)$.
A set~$U \subseteq V$ is a \emph{feedback vertex set} if~$G-U$ is acyclic (a \emph{forest}).
The minimum cardinality of a feedback vertex set of~$G$ is denoted~$\fvs(G)$. For a vertex $v$, denote by $N_G(v)$ the set of its neighbors (excluding $v$ itself).
The \emph{open} and \emph{closed} \emph{neighborhood} of a set~$U \subseteq V$ are~$N_G(U) := \bigcup _{s \in U} N_G(s) \setminus U$ and $N_G[U] := \bigcup _{s \in U} N_G(s) \cup U$, respectively. We omit the subscript when it is clear from the context. A graph $G'=(V',E')$ is said to be a \emph{subgraph} of $G$, if $V'\subseteq V$, and $E'\subseteq E$.
It is an \emph{induced subgraph} of $G$ if $V'\subseteq V$ and for any $u,v\in V'$ we have $\{u,v\} \in E$ if and only if $\{u,v\}\in E'$. The subgraph of~$G$ induced by~$U \subseteq V$ is denoted~$G[U]$. A \emph{graph class} is a (possibly infinite) collection of graphs.
A graph class $\Pi$ is said to be \emph{hereditary} if given any graph $G\in\Pi$, any induced subgraph of $G$ belongs to the class $\Pi$ as well.
A graph is \emph{bipartite} if its vertex set can be partitioned into two independent sets~$I$ and~$J$, which are also called the \emph{partite sets}. We sometimes denote such a bipartite graph by~$G = (I \cup J, E)$. A bipartite graph is \emph{balanced} if~$|I| = |J|$.
A \emph{matching} is a set of edges that do not share any endpoints. A matching \emph{covers} a vertex~$v$ if it contains an edge incident on~$v$. A matching is \emph{perfect} if it covers all vertices. We will utilize the following well-known consequence of K\H{o}nig's theorem.
\begin{fact}[{\cite[Corollary 16.7]{Schrijver03}}] \label{fact:konig}
Let~$G = (I \cup J, E)$ be a bipartite graph. Then~$G$ has a matching covering~$I$ if and only if~$|N(S)| \geq |S|$ for each~$S \subseteq I$.
\end{fact}
A vertex~$v$ is a \emph{cutvertex} in graph~$G$ if the removal of~$v$ increases the number of connected components. A graph is \emph{biconnected} if it does not contain a cutvertex. Under this definition, the graph~$K_2$ is biconnected. A \emph{biconnected component} of~$G$ is a maximal biconnected subgraph of~$G$.

\begin{definition}[{\cite[\S 7.2]{ParAlgo15}}]
A \emph{path decomposition} of a graph $G = (V,E)$ is a sequence $\mathcal{P} = (X_1, X_2,\ldots,X_r)$ of subsets of $V$ called \emph{bags}, satisfying the following conditions:
\begin{enumerate}[{\normalfont (P1)}]
\item $\bigcup_{i=1}^r X_i=V$. In other words, every vertex of $G$ is in at least one bag.
\item For every $\{u,v\}\in E$, there exists $l\in\{1,2,\ldots,r\}$ such that the bag $X_l$ contains both $u$ and $v$.
\item For every $v\in V$, if $u\in X_i\cap X_k$ for some $i\leq k$, then $u\in X_j$ also for each $j$ such that $i\leq j\leq k$. In other words, the indices of the bags containing $u$ form an interval in~$\{1,2,\ldots, r\}$.
\end{enumerate}
\end{definition}
The \emph{width} of a path decomposition $(X_1,\ldots, X_r)$ is $\max_{1\leq i\leq r}|X_i|-1$.
The \emph{pathwidth} of $G$, denoted by $\pw(G)$, is the minimum possible width of a path decomposition of $G$.
A path-decomposition $(X_1,X_2,\ldots,X_r)$ of a graph $G$ is \emph{nice} if the following holds:
\begin{enumerate}[{\normalfont (i)}]
\item $X_1=X_r=\emptyset$, and
\item for every $i\in\{1,2,\ldots, r-1\}$, there is either a vertex $v\notin X_i$ such that $X_{i+1}=X_i\cup\{v\}$, or there is a vertex $w\in X_i$ such that $X_{i+1}=X_i\setminus\{w\}$.
\end{enumerate}
It is well-known (cf.~\cite[Lemma 7.2]{ParAlgo15}) that every graph admits a nice path decomposition of width~$\pw(G)$.
 For any path decomposition $\mathcal{P}=(X_1,X_2,\ldots,X_r)$ of $G = (V,E)$, and any vertex $v\in V$, define $l_{\mathcal{P}}(v)=\min\{i:v\in X_i\}$ and $r_{\mathcal{P}}(v)=\max\{i:v\in X_i\}$, i.e.~$l_{\mathcal{P}}(v)$ and $r_{\mathcal{P}}(v)$ respectively denote the index of the first and last bag containing $v$. Note that if $\mathcal{P}$ is nice, then $l_{\mathcal{P}}(\cdot)$ and $r_{\mathcal{P}}(\cdot)$ are injective maps over the set of vertices.

\section{Definitions and Basic Facts for Reconfiguration} \label{sec:reconfig}
In this section we formally define the two notions of reconfiguration and establish some basic facts.

\subparagraph{Multiple Token Jump (MTJ).}
Given any two independent sets $I$ and $J$, with $|I|=|J|$, we say that $I$ can be \emph{$k$-MTJ reconfigured} to $J$, if there exists a finite sequence of independent sets $(I = W_0, W_1,W_2,\ldots,W_n, W_{n+1}=J)$ for some $n\geq 0$, such that for all $i \in \{0, \ldots, n+1\}$ the set~$W_i$ is independent in~$G$, $|W_i|=|I|=|J|$, and $|W_{i+1}\setminus W_i|\leq k$.
A step $W_i\to W_{i+1}$ in the reconfiguration process with~$|W_{i}\setminus W_{i+1}|=k$ is called a $k$-TJ move.
Given a graph~$G = (V,E)$, define $\MTJ(G,s)$ as the minimum value of $k$, such that any two independent sets of size $s$ in~$G$ can be $k$-MTJ reconfigured to each other. Now define $\MTJ(G):=\max_{1\leq s\leq |V|} \MTJ(G,s)$.
Our goal is to characterize the value of $\MTJ(G)$ in terms of certain parameters of the graph $G$. We call $\MTJ(G)$ the \emph{MTJ reconfiguration threshold} of the graph $G$. The MTJ reconfiguration threshold of a graph class $\Pi$ is defined as $\MTJ(\Pi):=\sup_{G\in\Pi} \MTJ(G)$.

\subparagraph{Token Addition Removal (TAR).}
Given any two independent sets $I$ and $J$, with $|I|=|J|$, we say that $I$ can be \emph{$k$-TAR reconfigured} to $J$, if there exists a finite sequence of independent sets $(I = W_0, W_1,W_2,\ldots,W_n, W_{n+1}=J)$ for some $n\geq 0$, such that~$W_i$ is independent in~$G$, $|I|-|W_i|\leq k$, and $|W_{i-1}\Delta W_{i}|\leq 1$ for all $i \in \{0, \ldots, n+1\}$.
We refer to the quantity $B_i:=|I|-|W_i|$ as the \emph{buffer size} at step $i$: the tokens that were on the initial independent set, and are not on the current independent set~$W_i$, are placed in the buffer.
Define $\TAR(G,s)$ to be the smallest buffer size~$k$ such that any two independent sets of size $s$ can be $k$-TAR reconfigured to each other. Define $\TAR(G):=\max_{1\leq s\leq |V|} \TAR(G,s)$. As before, we call $\TAR(G)$ the \emph{TAR reconfiguration threshold} of the graph $G$, and extend the same terminology to graph classes~$\Pi$ by defining $\TAR(\Pi):=\sup_{G\in\Pi} \TAR(G)$.

\subparagraph{Facts on Reconfiguration.}
Observe that for any graph~$G$, it holds that $\MTJ(G) = 1$ if and only if $\TAR(G)= 1$. In general, the TAR reconfiguration threshold is at most the MTJ reconfiguration threshold. Indeed, to see this, observe that each $k$-TJ move can be thought of as a sequence of $2k$ steps with maximum buffer size $k$. First, sequentially remove the $k$ vertices that are jumping away, placing their tokens in the buffer; then sequentially place the buffer tokens on the $k$ new vertices in the independent set.

\begin{proposition} \label{prop:balancedbip}
Let~$G$ be a graph with independent sets~$I$ and~$J$ of equal size. If~$I \setminus J$ can be $k$-TAR reconfigured (resp.\,$k$-MTJ reconfigured) to $J \setminus I$ in the graph~$G[I \Delta J]$, then~$I$ can be $k$-TAR reconfigured (resp.\,$k$-MTJ reconfigured) to~$J$ in~$G$.
\end{proposition}
\begin{proof}
Consider a sequence of independent sets~$(I \setminus J = W_0, \ldots, W_{n+1} = J \setminus I)$ in $G[I\Delta J]$ that reconfigures~$I \setminus J$ to~$J \setminus I$. Since~$I$ and~$J$ are independent in~$G$, no vertex of~$I \Delta J$ is adjacent to a vertex of~$I \cap J$. Hence~$W'_i := W_i \cup (I \cap J)$ is an independent set in~$G$ for all~$i$, and the sequence~$(W'_0, \ldots, W'_{n+1})$ reconfigures~$(I \setminus J) \cup (I \cap J) = I$ to~$(J \setminus I) \cup (I \cap J) = J$ in~$G$. The step size and buffer size of this sequence in~$G$ are not greater than the corresponding values for the sequence in~$G[I \Delta J]$, which completes the proof.
\end{proof}

Proposition~\ref{prop:balancedbip} shows that to upper-bound the TAR or MTJ reconfiguration threshold, it suffices to do so in balanced bipartite graphs where the source and target configurations are disjoint; note that~$G[I \Delta J]$ is balanced bipartite and~$I \setminus J$ and~$J \setminus I$ are disjoint. We will frequently exploit this in our proofs. For any graph class $\Pi$, let $\Pi_\bip$ denote the set of bipartite graphs in $\Pi$. The following proposition shows that the reconfiguration threshold of a hereditary graph class is determined by the behavior of the bipartite graphs in the class. Note that for hereditary classes~$\Pi$, the class~$\Pi_\bip$ is also hereditary.

\begin{proposition} \label{prop:bip}
For any hereditary graph class $\Pi$, we have $\MTJ(\Pi)=\MTJ(\Pi_\bip)$ and
$\TAR(\Pi)=\TAR(\Pi_\bip)$.
\end{proposition}
\begin{proof}
The definitions of the thresholds imply that $\MTJ(\Pi)\geq \MTJ(\Pi_\bip)$ and $\TAR(\Pi)\geq \TAR(\Pi_\bip)$, since~$\Pi \supseteq \Pi_\bip$. For the reverse direction, assume that the reconfiguration threshold of~$\Pi_\bip$ (in one of the models) is at most~$k$ and consider any graph~$G \in \Pi$ with independent sets~$I$ and~$J$ of equal size. By Proposition~\ref{prop:balancedbip} the cost of reconfiguring~$I$ to~$J$ is bounded by the cost of reconfiguring~$I \setminus J$ to~$J \setminus I$ in~$G[I \Delta J]$. Since~$G[I \Delta J]$ is bipartite and~$\Pi$ is hereditary, we have~$G[I \Delta J] \in \Pi_\bip$, and hence the cost of reconfiguring in~$G[I \Delta J]$ is at most~$k$. So reconfiguring~$I$ to~$J$ can be done with cost at most~$k$ in this model.
\end{proof}

\section{Threshold for Multiple Token Jump Reconfiguration} \label{sec:mtj}
We start our discussion of token jump reconfiguration by recalling the following known result.
\begin{theorem}[{\cite[Theorem~7]{KMM12}}] \label{thm:tree}
Let the graph $G=(V,E)$ be a forest. Then $\MTJ(G) \leq 1$.
\end{theorem}

The intuition behind this result is that since a forest does not contain any cycle, one can start reconfiguring from the leaf nodes or the isolated vertices, each of which has at most one neighbor from the target configuration. For arbitrary graphs, the above procedure does not work since there may not be any leaves or isolated vertices. But if a graph~$G$ has a small vertex cover, then its MTJ reconfiguration threshold is again small.
\begin{theorem} \label{thm:max-match}
Let $G=(V,E)$ be a graph. Then $\MTJ(G)\leq \max(\vc(G),1)$.
\end{theorem}
\begin{proof}
We prove the theorem using induction on the number~$n$ of vertices in $G$. For~$n=1$ the claim is trivially true, so consider a graph~$G = (V,E)$ with source and target independent sets~$I$ and~$J$ of equal size and~$|V| > 1$. Our induction hypothesis is that any graph~$G'$ with less than~$|V|$ vertices has MTJ reconfiguration threshold upper-bounded by~$\max(\vc(G'),1)$.

By Proposition~\ref{prop:balancedbip} it is enough to show that in the graph $G[I\Delta J]$ induced by $I\Delta J$, starting from $I'=I\setminus J$, one can construct a sequence of MTJ moves to reach the configuration $J'=J\setminus I$ with step-size at most $\max(\vc(G),1)$. Let~$V' := I \Delta J$. Note that~$\vc(G[I \Delta J]) \leq \vc(G)$, and let $S\subseteq V'$ be a vertex cover of $G[I \Delta J]$ of cardinality at most $\vc(G)$. If~$S$ is a vertex cover of~$G[I \Delta J]$, then there is no edge between any two vertices of $V' \setminus S$. Assume without loss of generality that $|I'\cap S|\geq |J'\cap S|$ (otherwise swap the role of $I'$ and $J'$, which does not affect reconfigurability). We distinguish three cases.

\textbf{Case 1.} If the vertex cover is empty ($S = \emptyset$), then~$G[I' \Delta J']$ has no edges. Consequently, all vertex subsets in the graph are independent, and we can reconfigure~$I'$ to~$J'$ by jumping one token at a time. By Proposition~\ref{prop:balancedbip}, this implies~$I$ can be reconfigured to~$J$ in~$G$ using jumps of size~$1$.

\textbf{Case 2.} Suppose that~$|I'| = |J'| \leq \vc(G)$. Then we can jump the tokens from~$I'$ onto~$J'$ in a single step of size at most~$\vc(G)$, and complete the argument using Proposition~\ref{prop:balancedbip}.

\textbf{Case 3.} If the previous cases do not apply, we claim that~$s := |I' \cap S| > 0$. Indeed, if~$s = 0$, then since~$|I' \cap S| \geq |J' \cap S|$ we would have~$I' \cap S = J' \cap S = \emptyset$, implying that~$S = \emptyset$ and that the first case applies. Moreover, we have~$|J'\setminus S|\geq s$, otherwise
\begin{equation}
|J'|=|J'\cap S|+|J'\setminus S|\leq (\vc(G)-s)+(s-1)=\vc(G)-1,
\end{equation}
and we are in the previous case. Now let~$Z$ be an arbitrary set of~$s$ vertices from~$J' \setminus S$.
Choose all the vertices in $I'\cap S$ and jump their tokens to~$Z$, i.e., remove the vertices $I'\cap S$ from the independent set~$I'$ and replace them by~$Z$ to obtain~$I''$. The set~$I''$ is independent because the fact that~$S$ is a vertex cover implies that the only neighbors of~$Z \subseteq J' \setminus S$ belong to the set~$S$, while~$I''$ contains no vertex from~$S$. Since~$|I \cap S| \leq |S| \leq \vc(G)$, the step size of this move is at most~$\vc(G)$.

Consider the graph~$G'$ which is obtained from~$G[I \Delta J]$ by removing~$Z$ and~$I' \cap S$, which is again a balanced bipartite graph. Note that since the only neighbors of~$Z$ belong to~$I$ (since~$G$ is bipartite and~$Z \subseteq J$), and belong to~$S$ (since~$S$ is a vertex cover and~$Z \cap S = \emptyset$), it follows that~$G'$ contains no vertex that is a neighbor of~$Z$ in~$G[I \Delta J]$. Consequently, the union of~$Z$ with any independent set in~$G'$ is independent in~$G[I \Delta J]$. Since~$G'$ is smaller than~$G$, by induction one can reconfigure~$I' \setminus S$ to~$J' \setminus Z$ in the graph~$G'$ with steps of size at most~$\vc(G') \leq \vc(G)$. Adding~$Z$ to each set in the corresponding reconfiguration sequence produces a sequence that reconfigures~$(I' \setminus S) \cup Z$ to~$J'$ in~$G[I \Delta J]$. By inserting the step from~$I'$ to~$(I' \setminus S) \cup Z$ at the front of this sequence, we obtain a reconfiguration from~$I'$ to~$J'$ with steps of size at most~$\vc(G)$ in~$G[I \Delta J]$. By Proposition~\ref{prop:balancedbip} this implies that~$I$ can be reconfigured to~$J$ with steps of size~$\vc(G)$ in~$G$, completing the proof for this case.
\end{proof}


An even cycle of length~$2n$ has MTJ reconfiguration threshold~$n$. Since its vertex cover number is~$n$, Theorem~\ref{thm:max-match} is best-possible. Long cycles are not the only graphs whose MTJ reconfiguration threshold equals half the size of the vertex set. Bistable graphs (defined below), of which the pumpkin structure defined in the introduction is a special case, also have this property.

\subsection{MTJ Reconfiguration Threshold in Terms of Bistable Rank} \label{sect:mtj:bistable}
In this section we introduce the notion of \emph{bistable graph}, derive several properties of bistable graphs, and use these to bound the MTJ reconfiguration threshold in terms of the size of the largest induced bistable subgraph. The resulting bounds on the MTJ reconfiguration threshold are tight, but can be hard to apply to specific graph classes: it may be difficult to estimate the size of the largest induced bistable graph, or even to determine whether a given graph is bistable or not. In Section~\ref{sect:mtj:pumpkin} we will therefore relate the size of the largest induced bistable subgraph to the size of the largest pumpkin subgraph. This will result in upper- and lower bounds on the MTJ reconfiguration threshold in terms of the largest pumpkin structure contained in the graph (class), which is arguably a more insightful parameter. The resulting bound will not be best-possible, however.

\begin{definition}[Bistable graphs] \label{def:bistable}
A graph is called \emph{bistable} if it is connected, bipartite, and has exactly two distinct maximum independent sets formed by the two partite sets in its unique bipartition. The \emph{rank} of a bistable graph is defined as the size of its maximum independent sets.

Let~$\bistable(G)$ denote the rank of the largest induced bistable subgraph of~$G$. If~$G$ contains no induced bistable subgraphs (which can only occur if~$G$ has no edges), then we define~$\bistable(G)$ to be one. For a graph class $\Pi$ we define $\bistable(\Pi):=\sup_{G\in\Pi}\bistable(G)$.
\end{definition}

The pumpkin shown in Figure~\ref{fig:pumpkin} forms an example of a bistable graph. Lemma~\ref{lemma:bistable} connects bistable graphs to independent set reconfiguration. Consider the task of reconfiguring the $J$-partite set to the $I$-partite set in a balanced bipartite graph~$G = (I \cup J, E)$. If we have a set~$S \subseteq I$ such that~$|S| \geq |N(S)|$, then one way to make progress in the reconfiguration is to select~$|S|$ vertices from~$N(S) \subseteq J$ and jump their tokens onto the vertices in~$S$, resulting in a new independent set of the same size. The following lemma shows that when we consider a set~$S$ that is \emph{minimal} with respect to being at least as large as its neighborhood, then the induced subgraph~$G[N[S]]$ is bistable. Hence the cost of such a jump of~$|S|$ vertices is bounded by~$\bistable(G)$, which will allow us to bound the MTJ reconfiguration threshold.

\begin{lemma} \label{lemma:bistable}
Let~$G = (I \cup J, E)$ be a balanced bipartite graph without isolated vertices and let~$S \subseteq I$ be inclusion-wise minimal with the properties that~$|S| \geq |N(S)|$ and~$S$ is not empty. Then~$G[N[S]]$ is bistable.
\end{lemma}
\begin{proof}
We have to show that the graph~$G'$ induced by the vertices from~$S$ and their neighborhood satisfies all conditions for being bistable. Since~$G$ is bipartite,~$G'$ is as well. Before proving the remaining properties, we establish the following claim.

\begin{claim} \label{claim:bistable:matching}
The graph~$G'$ has a matching covering~$S$.
\end{claim}
\begin{claimproof}
Assume for a contradiction that no such matching exists. By Fact~\ref{fact:konig}, there is a set~$S' \subseteq S$ with~$|N_G(S')| = |N_{G'}(S')| < |S'|$. As~$G$ has no isolated vertices, we have~$|S'| > 1$. Removing an arbitrary vertex from~$S'$ to obtain~$S''$ then decreases the size of the set by at most one without increasing the neighborhood size. Hence~$|N_G(S'')| \leq |S''|$, a contradiction to the minimality of~$S$.
\end{claimproof}

We now show that~$G'$ has all properties of a bistable graph.

\subparagraph{Connectivity.} Assume for a contradiction that~$G'$ is not connected. If~$G'$ has a connected component with vertex set~$C$ that contains at least as many $I$-vertices as $J$-vertices, then~$S' := C \cap I$ is a strict subset of~$S$ with~$|S'| \geq |C \cap J| \geq |N_{G'}(S')| = |N_G(S')|$, contradicting minimality of~$S$. Otherwise, all connected components of~$G'$ have strictly more~$J$-vertices than~$I$-vertices. Since the $J$-vertices in~$G'$ form the neighborhood of~$S$, this implies that~$|N_G(S)| > |S|$, a contradiction to the choice of~$S$. Hence~$G'$ is connected.

\subparagraph{Balance.} Since~$G'$ is connected, it has a unique bipartition and it is easy to verify that~$S$ is one of the partite sets: $G' = (S \cup J', E')$. Since there is a matching covering~$S$ (Claim~\ref{claim:bistable:matching}) and all matching partners of vertices in~$S$ are distinct and belong to~$J'$, we therefore have~$|J'| \geq |S|$. We have~$|J'| = |N_G(S)| \leq |S|$ by assumption on~$S$, establishing that~$|J'| = |S|$ which proves~$G'$ is balanced. This implies that a matching in~$G'$ that saturates~$S$ (which exists by Claim~\ref{claim:bistable:matching}) is in fact a perfect matching in~$G'$.

\subparagraph{Two maximum independent sets.} Assume for a contradiction that~$G'$ has at least three maximum independent sets. Then there is a maximum independent set in~$G$ that is not equal to either of the two partite sets~$J'$ or~$S$; let~$X$ be such a maximum independent set. Since~$G'$ is bipartite and has a perfect matching~$M$, the set~$X$ contains exactly one vertex from each matching edge in~$M$. Now let~$\hS := X \cap I$. Since~$X \neq J'$ by assumption, it follows that~$\hS$ is not empty; since~$X \neq S$, it is a proper subset of~$S$. We show that~$|N_G(\hS)| \leq |\hS|$, contradicting minimality of~$S$.

Let~$M' \subseteq M$ denote the matching edges intersected by~$\hS$. Since~$X$ contains one vertex from each edge of~$M$, for all edges in~$M \setminus M'$ the $J$-endpoint of the edge belongs to the independent set~$X$. So the $J$-endpoint of an edge in~$M \setminus M'$ is not in the neighborhood of~$\hS$, as~$X$ is independent. Consequently, only the matching partners of~$\hS$ can be in the neighborhood of~$\hS$, implying there are at most~$|\hS|$ such neighbors. Hence~$|N_G(\hS)| \leq |\hS|$; a contradiction. It follows that~$G'$ has at most two maximum independent sets. To see that it has exactly two, it suffices to observe that since~$G'$ has a perfect matching, both its partite sets are maximum independent sets.

This establishes that~$G'$ satisfies all conditions for being bistable and concludes the proof of Lemma~\ref{lemma:bistable}.
\end{proof}

Now, in the lemma below, we prove two key properties of bistable graphs. They will later be useful to relate the quantities~$\Pum(G)$ and~$\bistable(G)$.

\begin{lemma}\label{lem:bistable-property}
Let $G = (I\cup J,E)$ be a bistable graph. Then the following holds:
\begin{enumerate}
\item $G$ has a perfect matching covering $I$ (and hence $J$). \label{prop:bistable:pm}
\item $G$ is biconnected. \label{prop:bistable:biconnected}
\end{enumerate}
\end{lemma}
\begin{proof}
(\ref{prop:bistable:pm}) By K\H{o}nig's theorem (cf.~\cite[Thm. 16.2]{Schrijver03}), the size of a maximum matching in the bipartite graph~$G$ equals the size of a minimum vertex cover in~$G$. By Definition~\ref{def:bistable}, the partite sets~$I$ and~$J$ are maximum independent sets and therefore have equal size. Since the complement of a maximum independent set is a minimum vertex cover, it follows that~$V(G) \setminus I = J$ is a minimum vertex cover. Hence there is a matching of size~$|J| = |I|$ in~$G$, which is a perfect matching since it covers~$2|J| = |V(G)|$ vertices.

(\ref{prop:bistable:biconnected}) Assume for a contradiction that $G$ is not biconnected. Let $v$ be a cutvertex and let~$M$ be a perfect matching in~$G$, which exists by the previous property. Assume that~$v \in I$; the argument for~$v \in J$ is symmetric. Let~$u$ be the matching partner of~$v$ under~$M$. Since~$v$ is a cutvertex, the graph~$G - \{v\}$ consists of multiple connected components~$C_1, \ldots, C_\ell$. Without loss of generality, assume that~$u$ is contained in component~$C_1$. For all components~$C_i$ with~$i \geq 2$, the component contains the same number of~$I$ and~$J$-vertices: for each vertex its matching partner in the opposite partite set belongs to the same component. For component~$C_1$ the number of~$I$-vertices is one smaller than the number of $J$-vertices, since the matching partner of~$u$ does not belong to~$C_1$. Consider the set~$S$ consisting of the $J$-vertices from~$C_1$ along with the~$I$-vertices of all other components~$C_2, \ldots, C_\ell$ of~$G - \{v\}$. The set~$S$ is independent in~$G- \{v\}$ since it consists of entire partite sets of different components of the bipartite graph. Since~$v \not \in S$ it follows that~$S$ is also independent in~$G$. As~$S$ contains exactly one endpoint from each edge in~$M$ (the $I$-endpoint for matching edges intersecting a component~$C_i$ for~$i \geq 2$, and the $J$-endpoint for the remaining matching edges) it follows that~$S$ is a maximum independent set in~$G$ that differs from~$I$ and~$J$; a contradiction to Definition~\ref{def:bistable}.
\end{proof}

\begin{theorem} \label{thm:bistable}
For any graph~$G$ it holds that~$\MTJ(G) \leq \bistable(G)$. Moreover, if~$G \neq K_1$, then there exists an induced subgraph~$G'$ of~$G$ with~$\MTJ(G') \geq \bistable(G) \geq \bistable(G')$.
\end{theorem}
\begin{proof}
We first prove the lower bound on~$\MTJ(G)$.
Note that if~$G \neq K_1$ contains no induced bistable subgraphs, then $G$ is a collection of isolated vertices, and in that case $\MTJ(G)=\bistable(G)=1$.
Assume then that~$G$ contains a nonempty induced bistable subgraph
$G' = (I' \cup J', E')$ of rank~$\bistable(G)$.
By Definition~\ref{def:bistable}, the sets~$I'$ and~$J'$ are the only independent sets of size~$\bistable(G)$ in~$G'$. It follows that in any MTJ reconfiguration sequence from~$I'$ to~$J'$, the set~$I'$ is immediately followed by~$J'$ which requires a jump of~$|I'| = |J'| = \bistable(G)$ tokens simultaneously. Hence~$\MTJ(G') \geq \bistable(G)$.

We prove the upper bound on~$\MTJ(G)$ by induction on the size of the graph. If~$G$ consists of a single vertex, then there is a unique nonempty independent set, so~$\MTJ(G) = 0$.
In the remainder, assume~$G$ has more than one vertex and let~$I$ and~$J$ be two independent sets in~$G$ of equal size.
By Proposition~\ref{prop:balancedbip} it suffices to prove that~$I' := I \setminus J$ can be MTJ-reconfigured to~$J \setminus I$ in the graph~$G' := G[I \Delta J]$ with jumps of size at most~$\bistable(G)$.
Assume first that~$G'$ has no isolated vertices, and let~$S \subseteq I \setminus J$ be an inclusion-wise minimal nonempty subset of~$I \setminus J$ with the property that~$|S| \geq |N_{G'}(S)|$. Such a set exists since~$G'$ is a balanced bipartite graph with partite sets~$I \setminus J$ and~$J \setminus I$, so the set~$I \setminus J$ satisfies the stated condition (but may not yet be minimal).
It is easy to verify that since~$G'$ has no isolated vertices and~$S$ is minimal, we have~$|N_{G'}(S)| = |S|$. Now move all tokens from~$N_{G'}(S)$ onto~$S$ in a single jump of size~$|S|$. By Lemma~\ref{lemma:bistable}, the graph~$G'[N_{G'}[S]] = G[N_{G'}[S]]$ is a bistable induced subgraph of~$G$ of rank~$|S|$, and therefore~$\bistable(G) \geq |S|$ which shows that the size of the jump is sufficiently small.
We may then invoke induction similarly as in the proof of Theorem~\ref{thm:max-match} to complete the argument. If~$G'$ has an isolated vertex, then instead one can jump a token onto this isolated vertex and induct. This concludes the proof of Theorem~\ref{thm:bistable}.
\end{proof}

The following corollary characterizes the MTJ reconfiguration threshold of hereditary graph classes. It follows directly from Theorem~\ref{thm:bistable}. It applies to all graph classes except the one consisting only of the single graph~$K_1$ with a single vertex, for which the reconfiguration threshold is zero but~$\bistable(K_1) = 1$ by definition.

\begin{corollary}
For any hereditary graph class~$\Pi \neq \{K_1\}$ it holds that~$\MTJ(\Pi) = \bistable(\Pi)$.
\end{corollary}
\begin{proof}
By Theorem~\ref{thm:bistable} we have~$\MTJ(G) \leq \bistable(G) \leq \bistable(\Pi)$ for all graphs~$G \in \Pi$, hence $\MTJ(\Pi) \leq \bistable(\Pi)$. To prove the converse, consider an arbitrary graph~$G \in \Pi$ having at least one edge. Then~$G$ contains an induced bistable subgraph~$H = (I \cup J, E)$ of rank~$\bistable(G)$, and since~$\Pi$ is hereditary we have~$H \in \Pi$. Reconfiguring~$I$ to~$J$ in~$H$ requires a jump of size~$|I| = |J| = \bistable(G)$ since those are the only two independent sets of that size. Hence~$\MTJ(\Pi) \geq \bistable(G)$ for all~$G \in \Pi$ with at least one edge, showing that~$\MTJ(\Pi) \geq \bistable(\Pi)$ if~$\Pi$ contains at least one bistable graph. In the exceptional setting that~$\Pi$ contains no bistable graph, all graphs in~$\Pi$ are edgeless causing~$\bistable(\Pi)$ to be one. Since~$\Pi \neq \{K_1\}$, the graph consisting of two isolated vertices is contained in~$\Pi$, which has reconfiguration threshold one. Hence the lower bound also holds in this case.
\end{proof}

\subsection{MTJ Reconfiguration Threshold in Terms of Pumpkin Size} \label{sect:mtj:pumpkin}
In this section we formally introduce the pumpkin structure described in the introduction. We relate pumpkins to bistable graphs to obtain bounds on the MTJ reconfiguration threshold in terms of the size of the largest pumpkin subgraph.

\begin{definition}[{Pumpkin}]
A \emph{pumpkin} is a graph consisting of two terminal vertices~$u$ and~$v$ linked by two or more vertex-disjoint paths with an odd number of edges, having no edges or vertices other than those on the paths. A path can consist of the single edge~$\{u,v\}$. The \emph{size} of the pumpkin is the total number of vertices.

For a graph $G$ we denote by $\Pum(G)$ the size of the largest (not necessarily induced) subgraph isomorphic to a pumpkin that is contained in $G$, or zero if~$G$ contains no pumpkin. For a graph class $\Pi$ we define $\Pum(\Pi):=\sup_{G\in\Pi}\Pum(G)$.
\end{definition}
An example of a pumpkin structure is shown in Figure~\ref{fig:pumpkin}. Observe that a pumpkin is a bipartite graph, since all cycles consist of two $uv$-paths of odd length and are therefore even. Furthermore, a pumpkin is a \emph{balanced} bipartite graph: vertices~$u$ and~$v$ belong to different partite sets since their distance is odd, and on every (odd-length) $uv$-path in the structure there is an even number of interior vertices, which alternate between the two partite sets. It is not difficult to verify that the two partite sets are the only maximum independent sets in a pumpkin, leading to the following observation.

\begin{observation} \label{obs:pumpkin:bistable}
Every pumpkin graph is bistable.
\end{observation}

The next theorem shows that the rank of the largest bistable induced subgraph of~$G$ can be upper-bounded in terms of the size of~$G$'s largest pumpkin subgraph.

\begin{theorem}\label{thm:bistable-pumpkin}
For any bistable graph $G$, $\bistable(G)\leq f(\Pum(G))$, where $f(k)=(k^3+k^2)^{k^2+1}+1.$
\end{theorem}
\begin{proof}
Consider a bistable graph~$G = (I \cup J, E)$. If~$G$ is acyclic, then any biconnected subgraph of~$G$ contains at most two vertices. There is a unique bistable graph with at most two vertices, which consists of a single edge and has rank one. Since any bistable graph is biconnected by Lemma~\ref{lem:bistable-property}, we have $\bistable(G) \leq 1 = f(0) \leq f(\Pum(G))$ if~$G$ is acyclic. In the remainder we assume that~$G$ contains a cycle, which implies that~$\Pum(G) \geq 1$ since any cycle in the bipartite graph~$G$ is even and forms a pumpkin.

For ease of notation, define~$L := \Pum(G)$. Construct a depth-first search (DFS) tree~$T$ of $G$, starting at an arbitrary vertex~$r$ which becomes the root of the tree. By the structure of the DFS process, we obtain the following property: if~$u$ and~$v$ are vertices of~$G$ that are adjacent in~$G$, then~$u$ is an ancestor of~$v$ in~$T$, or~$v$ is an ancestor of~$u$. For~$v \in T$, we use~$T_v$ to denote the subtree of~$T$ rooted at~$v$. We will often use~$T_v$ to refer to the vertices in the tree as well.

\begin{claim} \label{claim:depth}
The depth of $T$ is at most $L^2$.
\end{claim}
\begin{claimproof}
Assume for a contradiction that there is a path from the root~$r$ of~$T$ to a leaf~$\ell$, consisting of more than~$L^2$ edges. By Lemma~\ref{lem:bistable-property}, graph~$G$ is biconnected. The existence of a path of more than~$L^2$ edges in a biconnected graph~$G$ is known~\cite[Theorem 1]{D52} to imply that~$G$ contains a simple cycle of length more than~$L$. Since~$G$ is bipartite the cycle is even and forms a pumpkin: it splits into two odd paths. So~$\Pum(G) > L$, a contradiction.
\end{claimproof}

\begin{figure}
\begin{subfigure}{5cm}
\begin{center}
\includegraphics[scale=0.35]{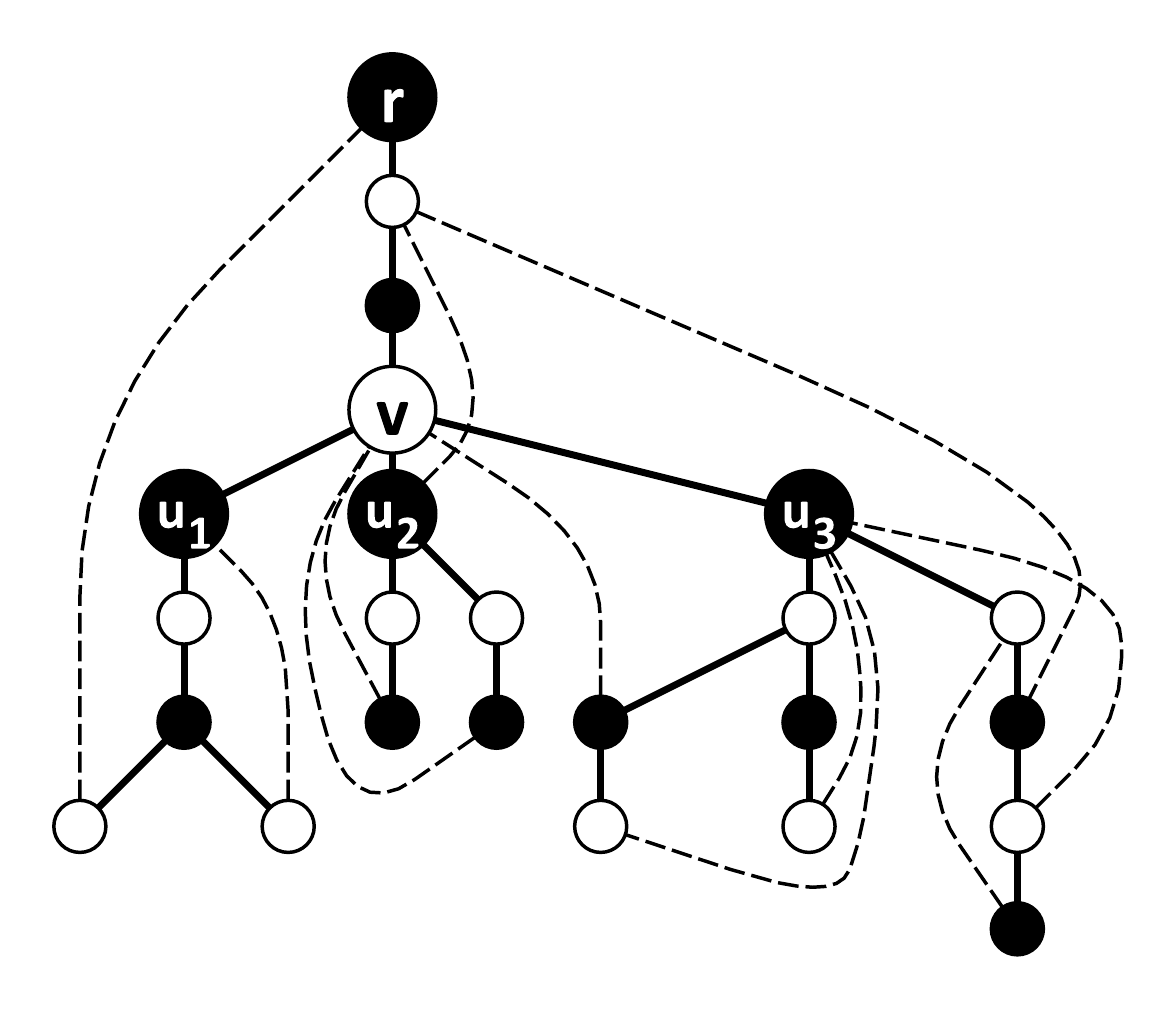}
\caption{DFS tree of a biconnected bipartite graph.}
\label{fig:DFS}
\end{center}
\end{subfigure}
\hspace{0.5cm}
\begin{subfigure}{8cm}
\begin{center}
\includegraphics[scale=0.35]{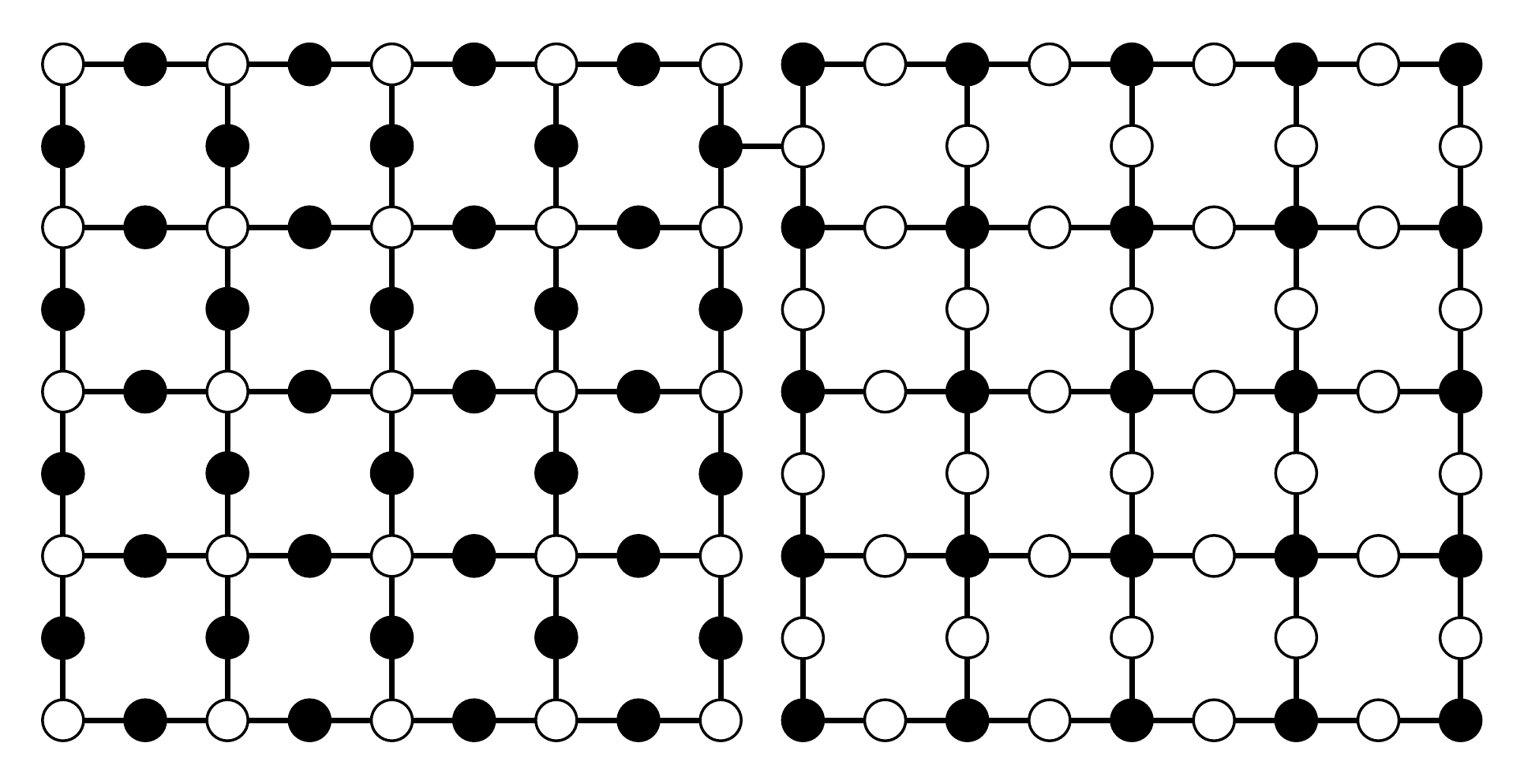}
\caption{Grid-like balanced bipartite graph with large treewidth and small TAR reconfiguration threshold.}
\label{fig:treewidth}
\end{center}
\end{subfigure}
\caption{(\ref{fig:DFS}) Depth-first search tree of a bipartite biconnected graph. Tree-edges are drawn solid, while the remaining edges of~$G$ are drawn with dotted lines. The three children~$u_1, u_2, u_3$ of~$v$ induce subtrees of types A, B, and~C, respectively. (\ref{fig:treewidth}) Template for constructing graphs of large treewidth that can be TAR reconfigured with a buffer of size two. The treewidth is large due to the presence of a large grid minor.}
\label{fig:DFS:treewidth}
\end{figure}

If each vertex in~$T$ has at most~$(L^3 + L^2)$ children, then the bound on the depth given by the previous claim implies that~$T$ (and therefore~$G$) has at most~$\sum _{i=0}^{L^2} (L^3 + L^2)^i = (L^3 + L^2)^{L^2+1} + 1$ vertices and therefore:
\begin{equation} \label{eq:pumpkinsize}
\bistable(G) \leq |V(G)| \leq (L^3 + L^2)^{L^2+1} + 1 \leq f(L) = f(\Pum(G)).
\end{equation}
To complete the proof, it therefore suffices to show that no vertex of~$T$ has more than~$L^3 + L^2$ children. Assume for a contradiction that some vertex~$v$ exists with a larger number of children~$u_1, \ldots, u_m$, for some~$m > L^3 + L^2$. By switching the roles of~$I$ and~$J$ if needed, we may assume that~$v \in I$. For a vertex~$w$, let~$M_w$ denote the set of its proper ancestors in~$T$. We classify the children $u_i$ for~$i \in [m]$ into three types:
\begin{enumerate}[{\normalfont Type A:}]
\item Some vertex of $T_u$ has an edge in $G$ to a vertex in $M_u \cap J$.
\item Vertex~$u$ is not of type A and $|I \cap T_u|\neq |J \cap T_u|$.
\item Vertex~$u$ is not of type A and $|I \cap T_u| = |J\cap T_u|$.
\end{enumerate}
Observe that any child of~$v$ belongs to exactly one type.

\begin{claim}
There are less than $L^3$ type-A children of~$v$.
\end{claim}
\begin{claimproof}
Suppose there are at least~$L^3$ type-A children of~$v$ and assume these are numbered~$u_1, \ldots, u_{L^3}$. Each subtree~$T_{u_i}$ for~$i \in [L^3]$ contains a vertex that has an edge in~$G$ to a proper ancestor of~$u_i$ in~$J$, by our definition of types. Since~$v \in I$, this cannot be~$v$ so that in fact it is also a proper ancestor of~$v$ and belongs to~$M_v$.

Since~$T$ has depth at most~$L^2$ by Claim~\ref{claim:depth}, by the pigeon-hole principle there is a vertex $w\in M_v \cap J$, such that $L$ subtrees among $T_{u_1},\ldots, T_{u_{L^3}}$ contain a vertex that is adjacent to~$w$ in~$G$. From each such subtree~$T_{u_i}$, we obtain a path in~$G$ from~$v$ to~$w$ whose internal vertices belong to~$T_{u_i}$, by going from~$v$ to~$u_i$, then to a neighbor of~$w$ in the subtree~$T_{u_i}$ using the tree edges, and ending with the edge to~$w$. Applying this procedure to each of the~$L$ subtrees that connect to~$w$ yields at least $L$ internally vertex-disjoint paths from $v$ to $w$. Since $G$ is bipartite and $v$ and $w$ belong to different partite sets, each path connecting $v$ and $w$ is of odd length. Hence this collection of~$L$ vertex-disjoint paths between~$v$ and~$w$ forms a pumpkin of size more than~$L$: each of the~$L$ paths has at least one internal vertex, and together with~$v$ and~$w$ this gives size at least~$L+2$. This contradicts our choice of~$L = \Pum(G)$.
\end{claimproof}

\begin{claim}
There are at most $\depth(v)+1$ type-B children of~$v$.
\end{claim}
\begin{claimproof}
Since $G$ is bistable, it has a perfect matching~$M$ by Lemma~\ref{lem:bistable-property}. By the properties of a DFS tree, for each vertex in~$T$ its neighbors in~$G$ are among its ancestors and descendants in~$T$. The number of $I$ and $J$-nodes in a subtree~$T_{u_i}$ rooted at a type-B child~$u_i$ are not equal. Since each vertex in~$T_{u_i}$ is assigned a unique neighbor in the other partite set by the perfect matching~$M$, it follows that the matching partner of some vertex in~$T_{u_i}$ does not belong to~$T_{u_i}$, and must therefore be a proper ancestor of~$u_i$ by the properties of DFS trees. Since there are at most~$\depth(v) + 1$ ancestors of~$v$ to use as matching partners, and each type-B child uses a different ancestor as a matching partner for one of its vertices, the number of type-B subtrees is at most~$\depth(v) + 1$.
\end{claimproof}

Since the depth of~$T$ is at most~$L^2$ by Claim~\ref{claim:depth}, any vertex~$v$ that is not a leaf has depth at most~$L^2 - 1$. Hence the number of type-B children of~$v$ is at most~$L^2$.

\begin{claim}
No child of~$v$ is of type~$C$.
\end{claim}
\begin{claimproof}
Suppose there exists a type-C child~$u_i$ of~$v$. Subtree~$T_{u_i}$ does not contain a vertex adjacent to~$M_{u_i} \cap J$, else~$u_i$ would have been type~$A$. Any $G$-neighbor of a vertex in~$T_{u_i}$ that is not contained in~$T_{u_i}$ is a proper ancestor of~$u_i$, by the properties of DFS trees. Hence the set $J' = (J \setminus T_{u_i}) \cup (I \cap T_{u_i})$ forms an independent set in~$G$. By the definition of type-C vertices, $|J \cap T_{u_i}|=|I \cap T_{u_i}|$, so that~$|J'| = |J|$. This shows that~$J'$ is a maximum independent set distinct from~$I$ and~$J$, contradicting the assumption that~$G$ is bistable.
\end{claimproof}

The preceding claims show that no vertex of~$T$ has more than~$L^2 + L^3$ children, which completes the proof of Theorem~\ref{thm:bistable-pumpkin} using (\ref{eq:pumpkinsize}).
\end{proof}

The following theorem is our main result on the MTJ reconfiguration threshold. It bounds the MTJ reconfiguration threshold of a hereditary graph class~$\Pi$ in terms of the maximum size of pumpkin subgraph of a graph in~$\Pi_\bip$. Recall that~$\Pi_\bip$ contains the bipartite graphs in~$\Pi$.

\begin{theorem}\label{th:pumpkin}
For any hereditary graph class $\Pi$, the following holds:
\begin{equation}
g_1(\Pum(\Pi_\bip))\leq \MTJ(\Pi)\leq g_2(\Pum(\Pi_\bip)),
\end{equation}
where $g_1$, $g_2:\N\to\N$ are positive non-decreasing functions defined as
$g_1(k)=k/2$ and~$g_2(k)=(k^3+k^2)^{k^2+1}+1$.
Moreover, for every graph~$G$ we have~$\MTJ(G)\leq g_2(\Pum(G))$.
\end{theorem}
\begin{proof}
We combine the bounds on the MTJ reconfiguration threshold of Theorem~\ref{thm:bistable}, with the relation between pumpkins and bistable graphs of Theorem~\ref{thm:bistable-pumpkin}.

\subparagraph{Lower bound on MTJ.} Consider a bipartite graph~$G$ in~$\Pi_\bip$. Then~$G$ contains a pumpkin subgraph on~$\Pum(G)$ vertices~$S \subseteq V(G)$. Then~$G[S] = (I \cup J, E)$ is a bipartite supergraph of a pumpkin, which is contained in~$\Pi_\bip$. Since any pumpkin is bistable by Observation~\ref{obs:pumpkin:bistable}, it follows that reconfiguring~$I$ to~$J$ in the pumpkin subgraph of~$G[S]$ requires a jump of size~$|I| = |J| = \Pum(G) / 2$. It is clearly no easier to reconfigure~$I$ to~$J$ in the supergraph~$G[S] \in \Pi_\bip$. Hence~$\MTJ(\Pi) \geq \Pum(G) / 2$ for all graphs~$G$ in~$\Pi_\bip$, giving the lower bound~$\MTJ(\Pi) \geq \Pum(\Pi_\bip) / 2$.

\subparagraph{Upper bound on MTJ.} By Proposition~\ref{prop:bip} it suffices to prove that~$\MTJ(\Pi_\bip) \leq g_2(\Pum(\Pi_\bip))$. Consider an arbitrary graph~$G \in \Pi_\bip$. By Theorem~\ref{thm:bistable}, we have~$\MTJ(G) \leq \bistable(G) \leq g_2(\Pum(G)) \leq g_2(\Pum(\Pi_\bip))$, where the second inequality follows from Theorem~\ref{thm:bistable-pumpkin}. Hence~$\MTJ(\Pi_\bip) \leq g_2(\Pum(\Pi_\bip))$, concluding the proof.
\end{proof}

While the upper bound of Theorem~\ref{th:pumpkin} has room for improvement, the following lemma shows that the exponential dependency on the pumpkin size in the upper bound is unavoidable.

\begin{proposition} \label{prop:super:pumpkin}
Let $\Pi_{\Pum}(k) := \{ G : \Pum(G) \leq k\}$ be the class of all graphs $G$ whose largest
pumpkin subgraph has size at most~$k$. Then $\MTJ(\Pi_{\Pum}(k)) = 2^{\Omega(k)}$.
\end{proposition}
\begin{proof}
For each~$k \geq 1$ we will construct a graph $G_k$ belonging to the class $\Pi_{\Pum}(24k+6)$
with $\MTJ(G_k) = 2^{\Omega(k)}$. We will call $G_k$ a \emph{super-pumpkin}. It is defined
recursively, as explained next.

Define a (4,2)-pumpkin, denoted $P_{4,2}$, to be a pumpkin whose two terminal
vertices are connected by four paths with two interior vertices each.
A super-pumpkin is now defined as follows. Like a regular
pumpkin, it has two designated terminal vertices.
The super-pumpkin $G_1$ consists of just a single edge, whose endpoints are its terminal vertices.
The super-pumpkin $G_k$ is obtained by gluing two copies of a
super-pumpkin~$G_{k-1}$---we will denote these copies by
$G_{k-1}^1$ and $G_{k-1}^2$---into a (4,2)-pumpkin~$P_{4,2}$.
This is done by identifying the terminal vertices of $G_{k-1}^1$ and $G_{k-1}^2$
with specific vertices of the (4,2)-pumpkin, as indicated in
Fig.~\ref{fi:super-pumpkin}.
\begin{figure}
\begin{center}
\includegraphics{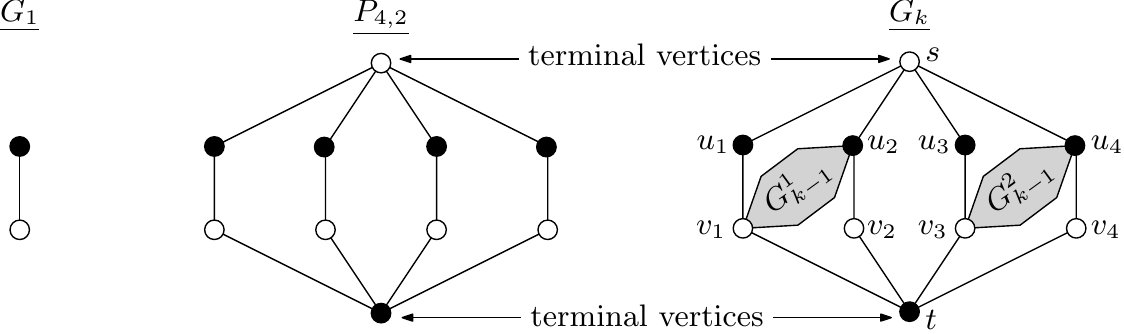}
\end{center}
\caption{Construction of a super-pumpkin.}
\label{fi:super-pumpkin}
\end{figure}

Note that $|G_k|$, the number of vertices of~$G_k$, satisfies $|G_k| = 2 |G_{k-1}| + 6$
with $|G_1|=2$. Hence, $|G_k|=2^k + 6\sum_{i=0}^{k-2} 2^i = 2^k + 6(2^{k-1}-1)$.
\begin{claim} \label{claim:super-size}
$G_k$ has exactly two independent sets of size~$|G_k|/2$, and these independent sets are disjoint.
\end{claim}
\begin{claimproof}
The proof is by induction on~$k$. It will be convenient to prove the following
stronger claim on $\I(G_k)$, the set of all independent sets of~$G_k$.
\begin{quotation}
\noindent $\I(G_k)$ contains no independent set of size more than~$|G_k|/2$ and exactly two independent sets
of size~$|G_k|/2$.  These independent sets are disjoint, and one of them contains one
terminal vertex of~$G_k$ while the other contains the other terminal vertex.
\end{quotation}
This claim trivially holds for $\I(G_1)$, so now consider $\I(G_k)$ for~$k>1$.
Let $s,t$ be the two terminal vertices of~$G_k$,
and label the other vertices of the pumpkin~$P_{4,2}$ as $u_1,\ldots,u_4$
and $v_1,\ldots,v_4$---see Fig.~\ref{fi:super-pumpkin}.
We define $W := \{s,t,u_1,u_3,v_2,v_4\}$ to be the set
of vertices in $G_k$ that do not occur in $G_{k-1}^1$ or $G_{k-1}^2$.
We distinguish two types of independent sets in~$\I(G_k)$.
\medskip

\noindent \textbf{Type~1.}
\emph{Independent sets $I$ such that both $G_{k-1}^1$ and $G_{k-1}^2$ have $|G_{k-1}|/2$ vertices in~$I$.} 
      Note that by the induction hypothesis we have $|\{u_2,v_1\}\cap I|= |\{u_4,v_3\}\cap I| =1$
      for any Type~1 independent set~$I$. Moreover, the total number of vertices from $I$
      inside $G_{k-1}^1$ and $G_{k-1}^2$ is $2(|G_{k-1}|/2)=|G_k|/2-3$.
      We will argue that $G_k$ has two Type~1 independent sets with $|G_k|/2$
      vertices and with the required properties, and that all other Type~1 independent sets
      have less than $|G_k|/2$ vertices. To this end
      we distinguish three subtypes of Type~1.
      \begin{itemize}
      \item \emph{Type~1(i). Independent sets $I$ with $u_2\in I$ and $u_4\in I$.}
            By the induction hypothesis such independent sets $I$ exist, and the choice of vertices of 
            $I$ inside $G^1_{k-1}$ and $G^2_{k-1}$ is fixed.
            Moreover, there is only one way to obtain an independent set $I^*$ with $|G_k|/2$
            vertices, namely by adding $\{u_1,u_3,t\}$ from~$W$---all other selections from~$W$
            give smaller independent sets.
      \item \emph{Type~1(ii). Independent sets $I$ with $v_1\in I$ and $v_3\in I$.}
            Again, the choice of vertices
            for~$I$ inside $G_{k-1}^1$ and $G_{k-1}^2$ is fixed, and there is
            only one way to obtain an independent set of $|G_k|/2$ vertices, this time by adding $\{s,v_2,v_4\}$.
            This independent set $I^{**}$ is disjoint from $I^*$---this follows from the induction hypothesis and the fact $\{u_1,u_3,t\}\cap\{s,v_2,v_4\}=\emptyset$---and it contains~$s$
            while $I^*$ contains~$t$.
      \item \emph{Type~1(iii). Independent sets $I$ with $u_2\in I$ and $v_3\in I$, or  $v_1\in I$ and $u_4\in I$.}
            Now at most two of the vertices from~$W$ can be in~$I$, and so $|I|<|G_k|/2$.
      \end{itemize}
\textbf{Type 2.}
 \emph{Independent sets $I$ such that at least one of $G_{k-1}^1$ and $G_{k-1}^2$ has less than
      $|G_{k-1}|/2$ vertices in~$I$.} 
      We will argue that all such independent sets have less than $|G_k|/2$ vertices.

      Assume without loss of generality that $G^1_{k-1}$
      has less than $|G_{k-1}|/2$ vertices in~$I$. If $G^2_{k-1}$ also has less than
      $|G_{k-1}|/2$ vertices in~$I$, then
      the total number of vertices from $I$ in $G_{k-1}^1$ and $G_{k-1}^2$
      is at most $2(|G_{k-1}|/2-1) = |G_{k}|/2 -5$, and since at most four vertices can be
      selected from~$W$ we have $|I|<|G_{k}|/2$.
      If $G^2_{k-1}$ has $|G_{k-1}|/2$ vertices in~$I$,
      then $|\{u_4,v_3\}\cap I| =1$. Assume without loss of generality that $u_4\in I$.
      Then $s$ and $v_4$ are not in~$I$. Since $v_2$ and $t$ cannot be both in $I$,
      we conclude that we can select at most three vertices from~$W$ into~$I$.
      This again implies that $|I|<|G_k|/2$.\\

Note that each $I\in\I(G_k)$ is of Type~1 or Type~2, since $G_{k-1}^1$ and $G_{k-1}^2$
cannot have more than $|G_{k-1}|/2$ vertices in~$I$ by the induction hypothesis.
This finishes the proof of the claim.
\end{claimproof}
Claim~\ref{claim:super-size} implies that $\MTJ(G_k) = |G_k|/2 = 2^{\Omega(k)}$: since there are only two independent
sets of size~$|G_k|/2$, say $I$ and $J$, and these are disjoint, the only way to go from
$I$ to $J$ is to remove all tokens from $I$ and place them  onto~$J$. Next we bound
the size of the largest pumpkin in $G_k$.
\begin{claim} \label{claim:super-size2}
$\Pum(G_k) \leq 24k+6.$
\end{claim}
\begin{claimproof}
Define $d_{\max}$ to be the maximum degree in $G_k$ and $C_k$ to be the
maximum length of any simple cycle in~$G_k$.
Then $\Pum(G_k) \leq d_{\max} \cdot C_k/2$.

The following statement is easy to prove by induction: the degree of the terminal vertices
in $G_k$ is four, and the maximum degree of any other vertex in $G_k$ is six.
Hence, $d_{\max}=6$ and so $\Pum(G_k) \leq 3 C_k$.

Next we argue that $C_k\leq 8k+2$. To this end, define $L_k$ to be the length
(measured in number of vertices) of  a longest simple path in $G_k$ that ends at
the two terminal vertices of~$G_k$. Then $L_k = L_{k-1}+4$ with $L_1=2$, and
thus $L_k = 4k-2$. We now prove that $C_k \leq 8k+2$ by induction on~$k$.
We have $C_1=0$, so the statement is true for $k=1$.
Now suppose~$k>1$. Let $\mathcal{C}$ be a simple cycle in $G_k$.
If $\mathcal{C}$ stays within one of the copies of $G_{k-1}$
we have $|\mathcal{C}| \leq C_{k-1}$ by induction. Otherwise the
maximum possible length for $\mathcal{C}$ is obtained by taking a longest
path from $u_2$ to $v_1$ in the first copy of $G_{k-1}$,
a longest path from $u_4$ to $v_3$ in the second copy, and connecting
them into a cycle using all six vertices in~$W$, where $W$ is defined as before. Hence,
\[
C_k \leq \max(C_{k-1}, 2 L_{k-1} +6) = \max(C_{k-1},8k+2).
\]
It follows that $C_k \leq 8k+2$. Hence, $\Pum(G_k) \leq 3C_k \leq 24k+6$.
\end{claimproof}
This concludes the proof of Proposition~\ref{prop:super:pumpkin}.
\end{proof}

\section{Threshold for Token Addition Removal Reconfiguration} \label{sec:tar}
In this section we study the model of token additional removal. First observe that when~$G$ is a forest, we have $\MTJ(G)\leq 1$ and therefore $\TAR(G)\leq 1$ as well.
Also, from Theorem~\ref{thm:max-match} we get $\TAR(G)\leq \max(\vc(G),1)$.
But the inequality $\TAR(G)\leq \MTJ(G)$ tells us nothing about the behavior of the TAR reconfiguration threshold when the MTJ reconfiguration threshold is large.
The next simple proposition immediately points towards this direction.
Indeed observe that a large pumpkin (with large MTJ reconfiguration threshold) can have a small feedback vertex set; this happens for even cycles, for example.
\begin{proposition} \label{prop:fvs}
Let $G=(V,E)$ be a graph. Then $\TAR(G)\leq \fvs(G)+1$.
\end{proposition}

\begin{proof}
Let~$G = (V,E)$ be a graph with a minimum feedback vertex set~$S \subseteq V$ of size~$k$, and let~$I, J \subseteq V$ be independent sets of equal size. By Proposition~\ref{prop:balancedbip} we can assume that~$V = I \cup J$ and~$I \cap J = \emptyset$. If~$|I| = |J| \leq k$, then it is trivial to reconfigure~$I$ to~$J$ with a buffer of size at most~$k$, by first moving all tokens from~$I$ into the buffer, and then onto~$J$. In the remainder we assume~$|I| = |J| > k$. Let~$S_I$ be a superset of~$I \cap S$ of size~$k$, and let~$S_J$ be a superset of~$J \cap S$ of size~$k$. Then the graph~$G' := G - (S_I \cup S_J)$ is a subgraph of~$G - S$ and is therefore acyclic since~$S$ is a feedback vertex set. By Theorem~\ref{thm:tree} it follows that~$I' := I \setminus S_I$ can be MTJ reconfigured to~$J' := J \setminus S_J$ in~$G'$ by jumps of size~$1$, which easily implies that~$I'$ can be TAR reconfigured to~$J'$ in~$G'$ using a buffer of size at most~$1$; let~$\mathcal{S}$ be a corresponding reconfiguration sequence. To reconfigure~$I$ to~$J$ in~$G$, start by removing the tokens from the~$k$ vertices in~$S_I$ and place them in the buffer. Then apply the reconfiguration sequence~$\mathcal{S}$ to reconfigure~$I'$ to~$J'$, using at most~$1$ extra buffer token. Finish by moving the~$k$ buffer tokens onto~$S_J$ to arrive at the independent set~$J' \cup S_J = J$.
\end{proof}

One can see that the above bound is tight, by considering the TAR reconfiguration threshold of a complete balanced bipartite graph. Indeed for $K_{n,n}$, the minimum size of a feedback vertex set is $n-1$, and one can see that in order to include any one of the vertices of the target independent set the reconfiguration must pass through the empty set. This shows that the TAR reconfiguration threshold is also~$n$.

\subsection{TAR Reconfiguration Threshold in Terms of Pathwidth}
As the main result of this section, we will show that the TAR reconfiguration threshold of a graph is upper-bounded in terms of its \emph{pathwidth}. Before proving that statement, we present a structural lemma about path decompositions that will be useful in the proof.

\begin{lemma} \label{lem:pathwidthprefix}
Let~$G = (I \cup J, E)$ be a bipartite graph with a nice path decomposition~$\mathcal{P} = (X_1, \ldots, X_r)$ of width~$k$. Let~$S \subseteq J$ such that~$|N(S)| \leq |S|$ while no non-empty subset of~$S$ has this property. If we order the vertices in~$S$ as~$i_1, \ldots, i_t$ such that~$r_\mathcal{P}(i_1) < r_\mathcal{P}(i_2) < \ldots < r_\mathcal{P}(i_t)$, then~$|N(\{i_1, \ldots, i_{t'}\})| < t' + k$ for all~$1 \leq t' \leq t$.
\end{lemma}

Intuitively, the lemma says the following. Suppose a set~$S \subseteq J$ is inclusion-wise minimal with respect to being no smaller than its neighborhood. Then ordering~$S$ according to the right endpoints of the intervals representing~$S$ in the path decomposition, we are guaranteed that every prefix of~$S$ has a fairly small neighborhood compared to its size: the neighborhood size exceeds the size of the prefix by less than the pathwidth. Note that since the lemma deals with bipartite graphs only, no vertex of~$S$ can belong to the neighborhood of any prefix of~$S$. The ordering of the vertices is uniquely defined since the path decomposition is nice. The bound of Lemma~\ref{lem:pathwidthprefix} is best-possible. Consider a complete bipartite graph~$K_{n,n}$, with pathwidth~$n$. In any optimal path decomposition, for~$t' = 1$ the first vertex in the ordering has a neighborhood of size~$n$ and so~$n < t' + n = 1 + n$, but a better bound is not possible.

\begin{proof}[Proof of Lemma \ref{lem:pathwidthprefix}]
First observe that in a graph with a path decomposition of width~$k=0$ there can be no edges. Then the only vertex-minimal set~$S$ satisfying the assumptions is an isolated vertex, for which the claim trivially holds. In the remainder we assume~$k \geq 1$. For $t'=t$ we have~$\{i_1,\ldots,i_{t'}\}=S$, and by assumption $|N(S)|\leq |S|$. So for~$t' = t$ the claim in the lemma holds trivially for any $k\geq 1$. Assume for a contradiction that there is some~$t' < t$ such that:
\begin{equation}\label{eq:cross}
|N(\{i_1,\ldots,i_{t'}\})|-t'\geq k.
\end{equation}
We partition~$T := S \cup N(S)$ into three disjoint subsets to derive some structural properties that will lead to a contradiction.
\begin{enumerate}[{\normalfont (i)}]
\item $T_1:= \{v\in S\cup N(S):r_\mathcal{P}(v)\leq r_\mathcal{P}(i_{t'})\}$, the set of all vertices in $S\cup N(S)$ that are not contained in any of the bags after the bag with index $r_\mathcal{P}(i_{t'})$.
\item $T_2:= \{v\in S\cup N(S):l_\mathcal{P}(v)> r_\mathcal{P}(i_{t'})\}$, the set of all vertices in $S\cup N(S)$ that are not contained in any of the bags before or including the bag with index $r_\mathcal{P}(i_{t'})$.
\item $T_3:=  \{v\in S\cup N(S):l_\mathcal{P}(v)\leq r_\mathcal{P}(i_{t'})<r_\mathcal{P}(v)\}$, the set of all vertices in $S\cup N(S)$ that are contained in some bags before or including the bag with index $r_\mathcal{P}(i_{t'})$ and also in some bag after it.
\end{enumerate}
Observe that~$(T_1 \cap S) \cup (T_2 \cap S) \cup (T_3 \cap S)$ is a partition of~$S$, and that~$T_1 \cap S = \{i_1, \ldots, i_{t'}\}$. 

\begin{claim} \label{claim:sthree}
$|T_3|\leq k$.
\end{claim}
\begin{claimproof}
From property (P3) in the definition of path decomposition we know that $T_3\subseteq X_{\ell}$ for $\ell = r_{\mathcal{P}}(i_{t'})$.
Now $|X_{\ell}|\leq k+1$ since the width of~$\mathcal{P}$ is at most~$k$, and we know that $i_{t'}\in X_\ell \setminus T_3$.  Therefore we have $|T_3|\leq |X_\ell\setminus\{i_{t'}\}|\leq k$.
\end{claimproof}

For the remainder of the proof we distinguish two cases.

\subparagraph{Case 1: $T_2 \cap S = \emptyset$.} Then~$(T_1 \cap S) \cup (T_3 \cap S)$ is a partition of~$S$. By Claim~\ref{claim:sthree} we have~$|T_3 \cap S| \leq k$, and therefore
\begin{equation} \label{eq:caseone:ssize}
|S| = |T_1 \cap S| + |T_3 \cap S| \leq |T_1 \cap S| + k = t' + k.
\end{equation}

\begin{claim} \label{claim:prefixneighbors:beforetprime}
In Case 1 we have $N(\{i_{t'+1}, \ldots, i_t\}) \subseteq N(\{i_1, \ldots, i_{t'}\})$.
\end{claim}
\begin{claimproof}
Assume for a contradiction that~$v \in N(\{i_{t'+1}, \ldots, i_t\}) \setminus N(\{i_1, \ldots, i_{t'}\})$.
By~(\ref{eq:cross}) we have~$|N(\{i_1,\ldots,i_{t'}\})|\geq t' + k$, and the existence of~$v$ shows that
$$|N(S)| = |N(\{i_1, \ldots, i_t\})| > |N(\{i_1, \ldots, i_{t'}\})| \geq t' + k \geq |S|,$$ by (\ref{eq:caseone:ssize}). But this contradicts the starting assumption that~$|N(S)| \leq |S|$.
\end{claimproof}

\begin{claim} \label{claim:caseone:three}
In Case 1 we have~$|T_3 \cap S| \geq k$, implying that~$T_3 \subseteq S$ and~$|T_3 \cap S| = k$.
\end{claim}
\begin{claimproof}
Suppose that~$|T_3 \cap S| < k$. Then:
\begin{align*}
|N(S)| &\geq |N(\{i_1, \ldots, i_{t'}\})| & \text{since $S \supseteq \{i_1, \ldots, i_{t'}\}$ and~$S$ is independent,} \\
&\geq k + t' & \text{by (\ref{eq:cross}),} \\
&= |T_1 \cap S| + k & \text{since~$|T_1 \cap S| = t'$,} \\
&> |T_1 \cap S| + |T_3 \cap S| & \text{by the assumption~$k > |T_3 \cap S|$,} \\
&= |S| & \text{since~$T_2 \cap S = \emptyset$,}
\end{align*}
contradicting the precondition to the lemma. It follows that~$|T_3 \cap S| \geq k$. Since~$|T_3 \cap S| \leq |T_3| \leq k$ by Claim~\ref{claim:sthree}, it follows that~$|T_3 \cap S| = k$ and that all vertices of~$T_3$ belong to~$S$.
\end{claimproof}

Let~$\ell := r_\mathcal{P}(i_{t'})$. Since the path decomposition is nice there is only one vertex (i.e.,~$i_{t'}$) that occurs in~$X_\ell$ but not after~$X_\ell$. So~$X_\ell = \{i_{t'}\} \cup T_3$, and Claim~\ref{claim:caseone:three} implies that no vertex of~$N(S)$ occurs in~$X_\ell$ since~$X_\ell = \{i_{t'}\} \cup T_3 \subseteq S$. Claim~\ref{claim:prefixneighbors:beforetprime} shows that all neighbors of~$i_{t'+1}, \ldots, i_t$ are also neighbor to some vertex of the prefix~$i_1, \ldots, i_{t'}$. Since~$i_1, \ldots, i_{t'}$ are ordered by increasing right endpoint of the intervals representing them in the decomposition, all neighbors of~$i_{t'+1}, \ldots, i_t$ therefore have to occur in a bag with index at most~$r_\mathcal{P}(i_{t'})$, and since~$X_{\ell}$ contains no vertex of~$N(S)$, by (P3) it follows that no vertex of~$N(S)$ occurs in a bag with index~$\ell$ or later. Since~$X_\ell = \{i_{t'}\} \cup T_3$ and~$|T_3 \cap S| = k$ by the previous claim, there are~$k + 1$ vertices in~$X_\ell$. Since the size difference of consecutive bags in a nice path decomposition is exactly one, and no bag has size more than~$k+1$ since the width is~$k$, it follows that~$X_{\ell - 1} = X_\ell \setminus \{v\}$ for some vertex~$v \in \{i_{t'}\} \cup T_3 \subseteq S$. Since no vertex of~$N(S)$ occurs in bag~$X_\ell$ or after, and~$v$ does not occur in~$X_{\ell - 1}$ or earlier, it follows that~$v$ does not occur in a bag together with a vertex of~$N(S)$. By the definition of path decomposition, this implies that~$v$ has no neighbor in~$N(S)$; since~$v \in S$ and~$S$ is an independent set (it is a subset of a partite set of a bipartite graph), this implies that~$v$ is an isolated vertex in~$G$. But since~$1 \leq t' < t = |S|$, the set~$S' := \{v\}$ is a nonempty strict subset of~$S$ for which~$0 = |N(S')| \leq |S'| = 1$, contradicting the precondition to the lemma. This concludes the proof of Case 1.

\subparagraph{Case 2: $T_2 \cap S \neq \emptyset$.} We continue the proof of Lemma~\ref{lem:pathwidthprefix} for the case that~$T_2 \cap S \neq \emptyset$. We will show that $T_2\cap S$ is a nonempty strict subset of $S$ with $|T_2\cap S|\geq |N(T_2\cap S)|$. This will contradict our assumption that $S$ is inclusion-wise minimal with the property that $|S|\geq |N(S)|$. Now let us denote $|T_3\cap I|=k_I$ and $|T_3\cap J|=k_J$.
Note that~$T_3 \cap J = T_3 \cap S$, and observe from Claim~\ref{claim:sthree} that
\begin{equation}\label{eq:temp3}
k_I+k_J\leq k.
\end{equation}
Recall from the choice of $r_\mathcal{P}(i_{t'})$ that $|T_1\cap S|= |\{i_1,\ldots,i_{t'}\}|=t'$. Since $S=(T_1\cup T_2 \cup T_3)\cap S$, and the $T_i$'s are mutually disjoint, we have:
\begin{align*}
|S| &= |T_1\cap S| + |T_2\cap S| + |T_3\cap S|\\
&= t'+ |T_2\cap S| + k_J.
\end{align*}
Therefore,
\begin{equation}\label{eq:S2}
|T_2\cap S|=|S|-k_J-t'.
\end{equation}
Also note that
\begin{equation}\label{eq:temp}
\begin{split}
|N(S)| &= |N ((T_1\cup T_2 \cup T_3)\cap S)|\\
&\geq |N((T_1\cup T_2)\cap S)|\\
& = |N(T_1\cap S)|+|N(T_2\cap S)| - |N(T_1\cap S)\cap N(T_2\cap S)|.
\end{split}
\end{equation}
Now observe that any vertex which is a neighbor of some vertex in $T_1\cap S$ and some vertex in $T_2\cap S$, must be both in some bag with index at most~$r_{\mathcal{P}}(i_{t'})$ (to meet~$T_1 \cap S$) and in some bag with index strictly more than $r_{\mathcal{P}}(i_{t'})$ (to meet~$T_2 \cap S$).
This implies that
$
N(T_1\cap S)\cap N(T_2\cap S)\subseteq T_3\cap I.
$
Therefore
\begin{equation}\label{eq:temp2}
|N(T_1\cap S)\cap N(T_2\cap S)|\leq |T_3\cap I|=k_I.
\end{equation}
Hence, \eqref{eq:temp} and \eqref{eq:temp2} yield
\begin{equation}\label{eq:neighbor}
|N(T_2\cap S)|\leq |N(S)|-|N(\{i_1,\ldots,i_{t'}\})|+k_I.
\end{equation}
Therefore, combining \eqref{eq:S2} and \eqref{eq:neighbor} we get
\begin{align*}
|T_2\cap S|-|N(T_2\cap S)|&\geq |S|-k_J-t'-|N(S)|+|N(\{i_1,\ldots,i_{t'}\})|-k_I\\
&= (|S|-|N(S)|)+(|N(\{i_1,\ldots,i_{t'}\})|-t')-(k_I+k_J)\\
&\geq 0 \quad\mbox{from our hypothesis about $S$, and Equation }\eqref{eq:cross}\mbox{ and }\eqref{eq:temp3}.
\end{align*}
So~$T_2 \cap S$ is a nonempty strict subset of~$S$ satisfying the key property, contradicting that~$S$ is inclusion-wise minimal. This completes the proof of Lemma~\ref{lem:pathwidthprefix}.
\end{proof}

Using Lemma~\ref{lem:pathwidthprefix} we bound the TAR reconfiguration threshold in terms of pathwidth.

\begin{theorem}\label{th:pathwidth}
Let $G = (V,E)$ be a graph. Then $\TAR(G)\leq \max(\pw(G),1)$.
\end{theorem}
\begin{proof}
We prove this theorem using induction on the number of vertices.
As before, it is enough to consider $G=(V,E)$ and assume that the initial and target independent sets $I$ and $J$ respectively are such that $|I|=|J|$, $I\cup J=V$ and $I\cap J=\emptyset$.
We will show that~$\pw(G)\leq k$ implies that~$\TAR(G)\leq k$, using induction on the number of vertices~$n$. For $n=1$, the statement is trivially true. Now fix any $k\geq 1$, and assume the induction hypothesis that any graph $G$ with $n$ vertices satisfying $\pw(G)\leq k$ has $\TAR(G)\leq k$.

Assume $G$ is a graph of $n+1$ vertices having pathwidth at most $k$. Let $S$ be an inclusion-minimal subset of $J$ for which $|S|\geq |N(S)|$. Such a set exists since $|J| = |I| \geq |N(J)|$. We will show that if we reconfigure the set $S$ in a suitable order by moving tokens from~$N(S)$ onto~$S$, then the buffer size will not grow beyond $k$. There are enough vertices in~$S$ to accommodate all tokens on~$N(S)$, and afterward we will invoke induction.

We first deal with a special case. If~$S = \{v\}$ is a singleton set, then it has degree at most one since~$|S| \geq |N(S)|$. Move the token from the neighbor~$u$ of~$v$ (or from an arbitrary vertex~$u$, if~$v$ has no neighbors) into the buffer, and then onto~$v$. By induction there exists a TAR reconfiguration from~$I \setminus \{u\}$ to~$J \setminus \{v\}$ in~$G - \{u,v\}$ using a buffer of size at most~$\max(\pw(G - \{u,v\}), 1) \leq \max(\pw(G), 1)$. When inserting the token move from~$u$ onto~$v$ at the beginning of this sequence, we get a TAR reconfiguration from~$I$ to~$J$ with the desired buffer size. In the remainder of the proof we can therefore assume~$|S| \geq 2$. This implies that~$|S| = |N(S)|$: if~$|S| > |N(S)|$ and~$|S| \geq 2$, then we can remove a vertex~$v$ from~$S$ to obtain~$|S \setminus \{v\}| \geq |N(S \setminus \{v\})|$ for the nonempty set~$S \setminus \{v\}$, contradicting minimality.

Let $\mathcal{P}=(X_1,X_2,\ldots,X_r)$ be a nice path decomposition of width at most $k$. If~$G$ has no edges, then~$S$ is a singleton set containing an isolated vertex. Since we already covered that case, we know~$G$ has at least one edge, so any path decomposition has width~$k \geq 1$.
Enumerate the vertices of $S$ as~$i_1,\ldots, i_m$ such that $r_{\mathcal{P}}(i_1) < \ldots < r_{\mathcal{P}}(i_m)$. Hence the vertices are ordered by increasing rightmost endpoint of the interval of bags containing it.

In order to describe the reconfiguration procedure we suitably group several TAR reconfiguration steps together as one step in the algorithm.
In particular, one reconfiguration step in the algorithm described below will consist of a run of successive removals of nodes, followed by a single node addition.

We use the notion of a \emph{buffer set} $B_t$ at the $t^{th}$ step of the reconfiguration, such that $|B_t|$ will correspond to the number of tokens in the buffer at any particular time, and $\max_t |B_t|+1$ will correspond to the maximum buffer size of the corresponding TAR reconfiguration sequence. The buffer set is a subset of vertices, showing where the tokens in the buffer came from. At time step $t=0$, define $W_0=I$ to be the independent set of vertices with a token, and let the buffer set $B_0$ be empty. We will define intermediate independent sets~$W_i$ and buffer sets~$B_i$ representing the grouped reconfiguration steps. The algorithm stops when $W_m$ contains all vertices in $S$; we will then invoke the induction hypothesis to finish the sequence. From the sequence~$(W_0, W_1, \ldots, W_m)$ one obtains a formal reconfiguration sequence as defined in Section~\ref{sec:reconfig} by inserting ``transitioning independent sets'' in between~$W_i$ and~$W_{i+1}$ for all~$i$. From~$W_i$, repeatedly remove one vertex until arriving at~$W_{i+1} \setminus W_i$, and then add the single vertex of~$W_{i+1} \setminus W_i$ to the resulting set. 

For $t\geq 1$, the transition from~$t-1$ to~$t$ is obtained as follows. Let~$u_t$ be an arbitrary vertex from $B_{t-1}\cup (N(i_t)\cap W_{t-1})$. Intuitively, at step~$t$ we take the token from~$u_t$ (in the buffer set or on a neighbor of $i_t$) and move it onto vertex~$i_t$, causing~$u_t$ to disappear from the buffer and adding~$i_t$ to the independent set. To ensure the resulting set is independent, tokens on neighbors of~$i_t$ are moved into the buffer beforehand.
Observe that the above step is valid only if $B_{t-1}\cup (N(i_t)\cap W_{t-1})$ is nonempty. Below in Claim~\ref{claim:buffer} we show that due to the choice of $S$, this is indeed the case for all $t\leq m$. Formally, we obtain the following:
\begin{algorithm*}[{Reconfiguring graphs with small pathwidth}]
Initialize with~$B_0 = \emptyset$ and~$W_0 = I$. We recursively define~$B_t$ and~$W_t$ for~$t \geq 1$.
\begin{enumerate}
\item The neighbors of~$i_t$ that have tokens (i.e.~that are in the current independent set) are removed from the previous independent set~$W_{t-1}$, making room to add~$i_t$ to the new independent set: $W_t=(W_{t-1}\setminus N(i_t))\cup \{i_t\}$.
\item The neighbors of~$i_t$ belonging to the previous independent set~$W_t$ move to the buffer, while~$u_t$ is removed from the buffer since its token has moved onto~$i_t$:
\begin{equation}
B_t=(B_{t-1}\cup (N(i_t)\cap W_{t-1}))\setminus \{u_t\}.
\end{equation}
\end{enumerate}
\end{algorithm*}
As mentioned earlier, a step from $W_t$ to $W_{t+1}$ can be thought as a sequence of successive removals of the nodes $N(i_{t+1})\cap W_{t}$, and then addition of the node $i_{t+1}$.
During this successive TAR reconfiguration sequence corresponding to the step $W_t$ to $W_{t+1}$, the maximum buffer size is given by $|B_{t+1}|+1$, since the buffer size will be~$|B_{t-1}\cup (N(i_t)\cap W_{t-1})|$ just before the buffer token from~$u_t$ is moved onto~$i_t$.
Therefore, the maximum buffer size in the entire TAR reconfiguration sequence starting from $W_0$ and ending at $W_m$ is given by $\max_{0\leq t\leq m} |B_t|+1$.
Also, at the end of the algorithm, all vertices from the set $S$ will be in the independent set, and no vertex in the buffer set.
This can be seen by observing the following.
Initially all tokens were on the vertices belonging to the set $N(S) \subseteq I$, since~$S \subseteq J$.
At each step of the algorithm essentially one token is selected from $N(S)$ as long as the number of such tokens is positive, and is placed on some vertex in $S$.
Now since $|S|\geq |N(S)|$, all the tokens in $N(S)$ must eventually exhaust before the algorithm terminates placing one token at each vertex of $S$.
For the validity of the above algorithm we claim the following, which in turn also characterizes the size of the buffer set at all intermediate time steps.
\begin{claim} \label{claim:buffer}
For all $1 \leq t\leq m$ we have that
$B_{t-1}\cup (N(i_t)\cap W_{t-1})$ is nonempty,
and that $|B_t|=|N(\{i_1,\ldots,i_{t}\})|-t$.
\end{claim}
\begin{claimproof} 
Suppose on the contrary that there exists $t'\leq m$, such that $B_{t'-1}\cup (N(i_{t'})\cap W_{t'-1})$ is empty for the first time. If~$t' = 1$, then~$B_{t'-1} \cup (N(i_{t'}) \cap W_{t'-1})$ is empty, and in particular~$N(i_{t'}) = \emptyset$, so that~$i_{t'} = i_1$ is an isolated vertex. But since~$|S| \geq 2$ by our argument above, it follows that~$S' = \{i_1\}$ is a nonempty strict subset with~$|S'| \geq |N(S')|$; a contradiction. So in the remainder we consider~$t' > 1$. We show that, for all $t< t'$, $|B_t|=|N(\{i_1,\ldots,i_{t}\})|-t$.
Using this, we prove that $2 \leq t' \leq m$ leads to a contradiction.

Observe that for any $t< t'$, after the $t^{th}$ step of the algorithm, the total number of distinct vertices that have been added to the buffer set is given by $|N(\{i_1,\ldots,i_{t}\})|$. Furthermore, for all $t''\leq t < t'$, the set $B_{t''-1}\cup (N(i_{t''})\cap W_{t''-1})$ has always been nonempty.
This implies that at each step, precisely one token has been removed from the buffer, thus reducing the size of the buffer set by moving a buffer token onto a vertex that is added to the independent set.
Therefore, in total $t$ times the size of the buffer set reduces by one.
Since initially the buffer set was empty, for any $t<t'$
we have $|B_t|=|N(\{i_1,\ldots,i_{t}\})|-t$.

Since we have assumed that $B_{t'-1}\cup (N(i_{t'})\cap W_{t'-1})$ is empty, we know $B_{t'-1}$ is empty, and therefore from the above argument
$|B_{t'-1}|=|N(\{i_1,\ldots,i_{t'-1}\})|-(t'-1)=0$.

Defining $S':=\{i_1,\ldots,i_{t'-1}\}\subsetneq S$, we have $|N(S')|\leq |S'|$. Since~$t' \geq 2$ the set~$S'$ is nonempty, contradicting the minimality of~$S$. This proves the first part of the claim. Since the buffer does not become empty until after step~$t$, the given argument then also proves the second part of the claim.
\end{claimproof}


Note that in particular $|B_m|=|N(\{i_1,\ldots,i_{m}\})|-m= |N(S)|-|S| = 0$; the buffer empties for the first time only after reconfiguring the whole set.

It remains to show that throughout the process the buffer size will not grow beyond~$k$, i.e.~$|B_t|\leq k-1$, for all $t\leq m$.
Claim~\ref{claim:buffer}~(ii) implies that
$\max_{t\leq m}|B_t|\geq k$ if and only if $\exists\ t \leq m$ such that
$|N(\{i_1,\ldots,i_{t}\})|-t\geq k$,
which is not possible due to Lemma~\ref{lem:pathwidthprefix}.
This then ensures that throughout the algorithm, the buffer size will never exceed $k$.

Since the buffer set empties out after reconfiguring the set $S$,
after the execution of the algorithm, $W_m\cap J=S$ and $W_m\cap I\subset V\setminus (S\cup N(S))$.
Now define $G'=G- (S\cup N(S))$, $I'=I\cap W_m$, and $J'=J\setminus S$. Observe that $G'$ has pathwidth at most $k$, and
$|I'|=|I\cap W_m|=|I|-|S|=|J'|$.
Furthermore, since $S$ is non-empty, $|V(G')|\leq n$. By the induction hypothesis, there exists a TAR reconfiguration sequence from~$I'$ to~$J'$ in~$G'$ using a buffer of size at most~$k$. Since $N(S)$ is not in $G'$, any independent set in $G'$ remains to be an independent set in $G$ when augmented with the set $S$. Therefore we can first apply the given reconfiguration from~$N(S)$ to~$S$, followed by the reconfiguration from~$I'$ to~$J'$, to reconfigure~$I$ to~$J$ with a buffer of size at most~$k$. 
\end{proof}

Observe by considering a complete balanced bipartite graph on $2n$ vertices $K_{n,n}$, that in general the above bound is tight. Indeed, from \cite{B98} we know that $K_{n,n}$ has pathwidth equal to $n$, and as explained earlier, the TAR reconfiguration threshold is also $n$.

\subsection{Obstructions to TAR Reconfigurability}
Having proved Theorem~\ref{th:pathwidth}, it is natural to ask whether pathwidth in some sense characterizes the TAR reconfiguration threshold: does large pathwidth of a graph imply that its TAR reconfiguration threshold is large? This is not the case: the pathwidth of a complete binary tree is proportional to its depth~\cite{KinnersleyL94}, but its reconfiguration threshold is one by Theorem~\ref{thm:tree}.

We now identify a graph structure which forces the TAR reconfiguration threshold to be large.
First we formally introduce the special type of minor illustrated in Figure~\ref{fig:binary}.

\begin{definition}[{Bipartite topological double minor}] \label{def:btd:minor}
Let $G = (I \cup J, E)$ be a bipartite graph and let~$H$ be an arbitrary graph. Then $H$ is a \emph{bipartite topological double minor} of $G$, if one can assign to every $v\in V(H)$ a subgraph $\varphi(v)$ of $G$, which is either an edge or an even cycle in $G$, and one can assign to each edge $e = \{u, v\} \in E(H)$ a pair of odd-length paths $\psi_1(e)$, $\psi_2(e)$ in $G$, such that the following holds:
\begin{itemize}
\item For any $u,v \in V(H)$ with $u\neq v$ the subgraphs $\varphi(u)$ and $\varphi(v)$ are vertex-disjoint.
\item For any $v \in V(H)$ no vertex of $\varphi(v)$ occurs as an interior vertex of a path $\psi_1(e)$ or $\psi_2(e)$, for any $e \in E(H)$.
\item	For any $e, e' \in E(H)$ the paths $\psi_1(e)$ and $\psi_2(e')$ are internally vertex-disjoint.
\item	For any $e = \{u,v\} \in E(H)$ the paths $\psi_1(e)$ and $\psi_2(e)$ both have one endpoint in $\varphi(v)$ and one endpoint in $\varphi(u)$.
\item	For any $v \in V(H)$ and edge $\{u,v\}\in E(H)$, the attachment points of~$\psi_1(e)$ and~$\psi_2(e)$ in~$\varphi(v)$ belong to different partite sets.
\end{itemize}
The triple~$(\varphi, \psi_1, \psi_2)$ is a \emph{BTD-minor model} of~$H$ in~$G$. For an edge~$e \in E(H)$ we define $\psi'_1(e), \psi'_2(e) \subseteq V(G)$ as the \emph{interior} vertices of the paths~$\psi_1(e)$ and~$\psi_2(e)$, which may be~$\emptyset$ if the path consists of a single edge.
\end{definition}

Intuitively, $H$ occurs as a bipartite topological double minor (or \emph{BTD-minor}) if each vertex of $H$ can be realized by an edge or even cycle, and every edge of $H$ can be realized by two odd-length paths that connect an $I$-vertex of $\varphi(v)$ to a $J$-vertex of $\varphi(u)$ and the other way around, in such a way that these structures are vertex-disjoint except for the attachment of paths to cycles. The definition easily extends to bipartite graphs whose bipartition is not given, since a BTD-minor is contained within a single connected component of the graph, which has a unique bipartition.

\begin{proposition} \label{prop:btd:minor}
Let~$G = (I \cup J, E)$ be a bipartite graph having a connected graph~$H$ as a BTD-minor model~$(\varphi, \psi_1, \psi_2)$, such that each vertex of~$G$ is in the image of~$\varphi$,~$\psi_1$, or~$\psi_2$. Then~$G$ has a perfect matching with~$|I| = |J|$ edges, and for any independent set~$W$ in~$G$:
\begin{enumerate}
	\item For each vertex~$v$ of~$H$ we have~$|W \cap \varphi(v)| \leq |\varphi(v)| / 2$.
	\item For each edge~$e$ of~$H$ and~$i \in \{1,2\}$ we have~$|W \cap \psi_i'(e)| \leq |\psi_i'(e)| / 2$.
\end{enumerate}
For a \emph{maximum} independent set~$W$, equality holds in all cases.
\end{proposition}
\begin{proof}
To see that~$G$ has a perfect matching, observe that each~$\varphi(v)$ for~$v \in V(H)$ is either an edge or an even cycle, which can be covered completely be a matching consisting of edges from~$\varphi(v)$. For each~$e \in E(H)$ and~$i \in \{1,2\}$ there is an even number of interior vertices on the path~$\psi_i(e)$, since the path has odd length. The interior vertices~$\psi'_i(e)$ can therefore also be covered completely by a matching of edges among~$\psi'_i(e)$. Since each vertex of~$G$ is in the image of~$\varphi$ or~$\psi_{1,2}$, the sets~$\varphi(v)$ together with sets~$\psi'_i(e)$ for~$e \in E(H)$ and~$i \in \{1,2\}$ cover~$V(G)$. Since these sets are vertex-disjoint by Definition~\ref{def:btd:minor}, they form a partition of~$G$. By the preceding argument, this implies~$G$ has a perfect matching~$M$ where no edge crosses the described partition of~$V(G)$. This matching has size~$|I| = |J|$ since~$G$ is bipartite.

Now we prove the two claimed properties. If~$W$ is an independent set, it contains at most one endpoint from each edge in~$M$ and therefore contains at most half the vertices of each~$\varphi(v \in V(H))$ and each~$\psi'_i(e \in E(H))$. Any independent set~$W$ achieving equality for all these sets has size~$|I| = |J|$ and is therefore maximum.
\end{proof}

For a bipartite graph~$G$, let~$\tree(G)$ denote the largest integer~$k$ for which~$G$ contains a complete binary tree of depth~$k$ as a BTD-minor. For a class of bipartite graphs $\Pi$ we define $\tree(\Pi):=\sup_{G\in\Pi}\tree(G)$.

\begin{theorem}\label{th:bip minor}
There exists a real constant~$c > 0$ such that any hereditary graph class $\Pi$ satisfies~$\TAR(\Pi) \geq c \cdot \tree(\Pi_\bip)$.
\end{theorem}
\begin{proof}
As before, we consider a balanced bipartite graph $G\in\Pi_\bip$ with bipartition $V(G)=I\cup J$ that has a complete binary tree $T$ of depth $d$ as a BTD-minor.
Since the graph class is hereditary, for the lower bound, we consider only the subgraph of $G$ induced by $\bigcup_{v\in V(T)}\{\varphi(v)\} \cup \left(\bigcup_{e\in E(T)}\{\psi_1(e)\cup\psi_2(e)\} \right)$, and without loss of generality, we shall refer to it as $G$ itself.
\begin{fact}[\cite{BC09}] \label{fact:ndb} 
There is a universal constant $c_1>0$ such that if~$T$ is a complete binary tree of depth $d$, then
$\displaystyle \max_{1\leq i\leq |V(T)|}\min_{S\subseteq V(T);|S|=i}|N_T(S)|\geq c_1 \cdot d$.
\end{fact}
The above implies that there exists $i_0 \leq |V(T)|$, such that any size-$i_0$ subset of $V(T)$ has a neighborhood of size at least $c_1 \cdot d$. Let~$I \cup J$ be the unique bipartition of the connected graph~$G$, and consider an arbitrary TAR reconfiguration sequence from $I$ and $J$. In this sequence $(I = W_0, W_1, \ldots, W_t = J)$ of independent sets in $G$, look at the reconfiguration step when for the first time there exists $S\subseteq V(T)$ with $|S|=i_0$, such that the intermediate independent set $W$ at that step contains $\bigcup_{v\in S}(\varphi(v)\cap J)$, and for all $v\notin S$ it satisfies $(\varphi(v)\cap W\cap J)\subsetneq(\varphi(v)\cap J)$. We will prove that~$|J| - |W| \geq c_1 \cdot d$, implying that from the initial independent set of~$|I| = |J|$ tokens, at least~$c_1 \cdot d$ tokens must reside in the buffer.

To prove the theorem, consider the intermediate independent set $W$, and the set $S\subseteq V(T)$ with $|S|=i_0$  satisfying the above criteria. The following claim shows that for each vertex in~$N_T(S)$, the independent set~$W$ uses at least one vertex fewer than the maximum independent set~$J$ does.
\begin{claim}\label{claim:fewer}
Consider an edge $e=\{u,v\}\in E(T)$ with $u\in S$ and $v\notin S$, and let~$Q_{e,v} \subseteq V(G)$ denote the vertices in~$\varphi(v) \cup \psi'_1(e) \cup \psi'_2(e)$. The following holds:
\begin{equation}\label{eq:object}
|W\cap Q_{e,v}| < |J \cap Q_{e,v}|=\frac{|Q_{e,v}|}{2}.
\end{equation}
\end{claim}
\begin{claimproof}
By Proposition~\ref{prop:btd:minor}, the maximum independent set~$J$ contains exactly half the vertices of~$Q_{e,v}$.
If $|W\cap\psi'_i(e)|<|\psi'_i(e)|/2$ for some~$i \in \{1,2\}$, then we are done: by Proposition~\ref{prop:btd:minor} the set~$W$ contains fewer vertices from~$\psi'_i(e)$ that the maximum independent set~$J$ does, and this cannot be compensated within the other parts of the structure since~$J$ contains half the vertices there and no independent set contains more.
In the remainder, we can assume that~$W$ contains exactly half the vertices from~$\psi'_1(e)$ and~$\psi'_2(e)$. Then the following are true:
\begin{enumerate}[{\normalfont (i)}]
\item All $J$-nodes of $\varphi(u)$ are in $W$ (by our choice of~$W$ and since~$u \in S$).
\item Some $J$-node of $\varphi(v)$ is not in $W$ (by our choice of~$W$ and since~$v \not \in S$).
\item Some $I$-node of $\varphi(v)$ is not in~$W$. To see this, let~$i \in \{1,2\}$ such that~$\psi_i(e)$ is an odd-length path from a $J$-node in~$\varphi(u)$ to an $I$-node in~$\varphi(v)$, which exists by Definition~\ref{def:btd:minor}, and orient it in that direction. Since the first vertex on the path is a $J$-node in~$\varphi(u)$, it is contained in~$W$ as shown above. Hence the second vertex on the path, the first interior vertex, is not in~$W$. Since exactly half the interior vertices from~$\psi_i(e)$ belong to~$W$, every other interior vertex from~$\psi_i(e)$ is in~$W$. Since the path has an even number of interior vertices and the first interior vertex is not in~$W$, the last interior vertex must be in~$W$. But this prevents its $I$-node neighbor in~$\varphi(v)$ from being in~$W$.
\end{enumerate}
Therefore, since $\varphi(v)$ is either an edge or an even cycle, we have $|W \cap \varphi(v)|<|\varphi (v)|/2$ by observing the following: the only independent sets in~$\varphi(v)$ of size~$|\varphi(v)| / 2$ are~$\varphi(v) \cap I$ and~$\varphi(v) \cap J$, but~$\varphi(v) \cap W$ is not equal to either of these sets since it avoids a $J$-node and an $I$-node. Hence $|W \cap \varphi(v)|<|\varphi (v)|/2 = |J \cap \varphi(v)|$, and Proposition~\ref{prop:btd:minor} shows that this cannot be compensated in other parts of the minor model, implying~$|W \cap Q_{e,v}| < |J \cap Q_{e,v}|$.
\end{claimproof}

Using Claim~\ref{claim:fewer} we finish the proof of Theorem~\ref{th:bip minor}. For each~$v \in N_T(S)$, pick an edge~$e = \{u,v\}$ such that~$u \in S$. By Claim~\ref{claim:fewer} the set~$W$ contains less than half the vertices of~$Q_{e,v}$, while the maximum independent set~$J$ contains exactly half. Since the sets~$Q_{e,v}$ considered for different vertices~$v \in N_T(S)$ are disjoint, while Proposition~\ref{prop:btd:minor} shows that from the other pieces of the minor model~$W$ cannot use more vertices than~$J$ does, it follows that~$|W| \leq |J| - |N_T(S)| \leq |J| - c_1 \cdot d$. Hence the buffer contains at least~$c_1 \cdot d$ tokens.
\end{proof}

\section{Conclusion} \label{sec:conclusion}
In this paper we considered two types of reconfiguration rules for independent set, involving simultaneously jumping tokens and reconfiguration with a buffer. For both models, we derived tight bounds on the corresponding reconfiguration thresholds in terms of several graph parameters like the minimum vertex cover size, the minimum feedback vertex set size, and the pathwidth.
Many results in the literature concerning the parameter pathwidth can be extended to hold for the parameter treewidth as well.
This is not the case here; the upper bound on the TAR reconfiguration threshold in terms of pathwidth (Theorem~\ref{th:pathwidth}) cannot be strengthened to treewidth, since one can make arbitrarily deep complete binary trees as BTD-minors in bipartite graphs of treewidth only two (see Figure~\ref{fig:binary}).
On the other hand, there are bipartite graphs of large treewidth with TAR reconfiguration threshold two (Figure~\ref{fig:treewidth}). To characterize the TAR reconfiguration threshold one therefore needs to combine graph connectivity (as measured by the width parameters) with notions that constrain the parity of the connections in the graph. This is precisely why we introduced BTD-minors.
We conjecture that the converse of Theorem~\ref{th:bip minor} holds, in the sense that any hereditary graph class having a large TAR reconfiguration threshold must contain a graph having a complete binary tree of large depth as a BTD-minor. Our belief is based partially on the fact that a BTD-minor model of a deep complete binary tree is arguably the simplest graph of large pathwidth and feedback vertex number. Resolving this conjecture is our main open problem.
%
%

{\small
\subparagraph*{Acknowledgments.}  

This research was financially supported by The Netherlands Organization for Scientific Research (NWO) through TOP-GO grant 613.001.012, Gravitation Networks grant 024.002.003 and a Veni grant `Frontiers in Parameterized Preprocessing'. Debankur thanks Sem Borst for his comments on the motivation for the reconfiguration thresholds.

\bibliographystyle{plainurl}
\bibliography{reference}}

\end{document}